\documentclass[% jmp,%
 amsmath,amssymb,
]{revtex4-1}

\pdfoutput=1
\usepackage{comment}
\usepackage[utf8]{inputenc}
\usepackage{natbib}
\usepackage{varwidth}
\usepackage[abs]{overpic}
\usepackage{graphicx}% Include figure files
\usepackage{dcolumn}% Align table columns on decimal point
\usepackage{bm}% bold math
\usepackage{float}
\usepackage{color}

\usepackage{qcircuit}
\usepackage{physics}
\usepackage{algorithm}
\usepackage{mathtools}
\usepackage{wrapfig}
\usepackage{dblfloatfix}  
\usepackage[margin=1cm,font=small,labelfont=bf,
   justification=justified,
   format=plain]{caption}
\usepackage{hyperref}

\newcounter{algsubstate}
\renewcommand{\thealgsubstate}{\alph{algsubstate}}

\usepackage[dvipsnames]{xcolor}
\usepackage{algcompatible}
\usepackage{algpseudocode}
\usepackage{tikz}
\usetikzlibrary{calc}
\usepackage{twoopt}
\usepackage[export]{adjustbox}

\usetikzlibrary{fit,tikzmark}

\newcommand\drawCodeBox[2]{%
  \begin{tikzpicture}[remember picture,overlay]
    \coordinate (start) at ([yshift=1.7ex]pic cs:#1);
    \coordinate (end) at ([yshift=-0.3ex]pic cs:#2);
    \node[inner sep=2pt,draw=red,fit=(start) (end)] {};
  \end{tikzpicture}%
}

\makeatletter
\newcommand{\algmargin}{\the\ALG@thistlm}
\makeatother
\newlength{\whilewidth}
\algnewcommand{\parState}[1]{\State%
  \parbox[t]{\dimexpr\linewidth-3\algmargin}{\strut #1\strut}}
\algnewcommand\algorithmicinput{\textbf{INPUT:}}
\algnewcommand\INPUT{\item[\algorithmicinput]}

\algnewcommand\algorithmicoutput{\textbf{OUTPUT:}}
\algnewcommand\OUTPUT{\item[\algorithmicoutput]}

\makeatletter
\newcommand{\skipitems}[1]{%
  \addtocounter{\@enumctr}{#1}%
}

\makeatother

\makeatletter
\tikzset{%
  remember picture with id/.style={%
    remember picture,
    overlay,
    save picture id=#1,
  },
  save picture id/.code={%
    \edef\pgf@temp{#1}%
    \immediate\write\pgfutil@auxout{%
      \noexpand\savepointas{\pgf@temp}{\pgfpictureid}}%
  },
  if picture id/.code args={#1#2#3}{%
    \@ifundefined{save@pt@#1}{%
      \pgfkeysalso{#3}%
    }{
      \pgfkeysalso{#2}%
    }
  }
}

\def\savepointas#1#2{%
  \expandafter\gdef\csname save@pt@#1\endcsname{#2}%
}

\def\tmk@labeldef#1,#2\@nil{%
  \def\tmk@label{#1}%
  \def\tmk@def{#2}%
}

\tikzdeclarecoordinatesystem{pic}{%
  \pgfutil@in@,{#1}%
  \ifpgfutil@in@%
    \tmk@labeldef#1\@nil
  \else
    \tmk@labeldef#1,(0pt,0pt)\@nil
  \fi
  \@ifundefined{save@pt@\tmk@label}{%
    \tikz@scan@one@point\pgfutil@firstofone\tmk@def
  }{%
  \pgfsys@getposition{\csname save@pt@\tmk@label\endcsname}\save@orig@pic%
  \pgfsys@getposition{\pgfpictureid}\save@this@pic%
  \pgf@process{\pgfpointorigin\save@this@pic}%
  \pgf@xa=\pgf@x
  \pgf@ya=\pgf@y
  \pgf@process{\pgfpointorigin\save@orig@pic}%
  \advance\pgf@x by -\pgf@xa
  \advance\pgf@y by -\pgf@ya
  }%
}

\makeatother
\newcounter{mymark}

\usepackage{subcaption}
\usepackage{cleveref}
\usepackage{blindtext}

\usepackage{amsthm}
\theoremstyle{plain}

\usepackage{enumitem}

\newtheorem{defn}{Definition}
\newtheorem{thm}{Theorem}
\newtheorem*{thm*}{Theorem}
\newtheorem{lem}{Lemma}
\newtheorem{meas}{Measurement}

\crefformat{thm}{#2#1#3}
\crefformat{lem}{#2#1#3}
\crefformat{meas}{#2#1#3}
\crefformat{figure}{#2#1#3}
\crefformat{equation}{(#2#1#3)}

\newcommand{\eo}{\mathrm{EO}}
\newcommand{\poly}{\mathrm{poly}}
\newcommand{\ii}{\mathrm{i}}
\newcommand{\ee}{\mathrm{e}}

\newcommand{\ur}{U_{\mathrm{radix}}}

\theoremstyle{definition}

\begin{document}
\setlength{\parindent}{1.0em}

\title[]{Graph comparison via nonlinear quantum search}
%\thanks{Footnote to title of article.}

\author{M. Chiew, K. de Lacy, C. H. Yu, S. Marsh, J. B. Wang}
\affiliation{Physics Department, The University of Western Australia, Perth, WA 6009, Australia}

\date{\today}% It is always \today, today,

\begin{abstract}

In this paper we present an efficiently scaling quantum algorithm which finds the size of the maximum common edge subgraph for a pair of arbitrary graphs and thus provides a meaningful measure of graph similarity. The algorithm makes use of a two-part quantum dynamic process: in the first part we obtain information crucial for the comparison of two graphs through linear quantum computation. However, this information is hidden in the quantum system with vanishingly small amplitude that even quantum algorithms such as Grover's search are not fast enough to distill the information efficiently. In order to extract the information we call upon techniques in nonlinear quantum computing to provide the speed-up necessary for an efficient algorithm. The linear quantum circuit requires $\mathcal{O}(n^3 \log^3 (n) \log \log (n))$ elementary quantum gates and the nonlinear evolution under the Gross-Pitaevskii equation has a time scaling of $\mathcal{O}(\frac{1}{g} n^2 \log^3 (n) \log \log (n))$, where $n$ is the number of vertices in each graph and $g$ is the strength of the Gross-Pitaveskii non-linearity. Through this example, we demonstrate the power of nonlinear quantum search techniques to solve a subset of NP-hard problems.

\begin{comment}
Graph comparison is the task of identifying the structural similarities between two networks. The maximum edge overlap (MEO) between two graphs provides a meaningful measure of graph similarity, and is exponentially difficult to compute for a pair of any two known graphs with classical algorithms, in terms of the number of graph vertices. We present an efficiently scaling quantum algorithm that finds the MEO between a pair of general graphs, achieving exponential speedup over classical methods.

The algorithm makes use of a two-part quantum dynamic process: in the first part we encode information crucial for the comparison of the two graphs in a single qubit. This information is hidden in a vanishingly small component of the system, having amplitude as small as $\mathcal{O}(1/n!)$ (where $n$ is the number of vertices in the larger graph). Because of this, even quantum algorithms such as Grover's search are not fast enough to distil this component efficiently. In order to extract the information we call upon techniques in nonlinear quantum computing to provide sufficient speed-up.

Our overall quantum algorithm finds the MEO of two graphs, each with at most $n$ vertices. We employ $\mathcal{O}(n^3 \log^3(n) \log\log(n))$ fundamental quantum gates and nonlinear evolution for time $\mathcal{O}(\frac{1}{g} n^2 \log^3(n) \log\log(n))$, where $g$ is the coefficient of nonlinearity in our nonlinear quantum system. Amongst the creation of new quantum subroutines, our results demonstrate the power of nonlinear quantum search techniques.
\end{comment}

\end{abstract}

%Use showkeys class option if keyword
\maketitle

\section{\label{sec:level1}Introduction}

Graph comparison is the task of quantifying the structural similarities between two graph topologies. Possession of a scalable algorithm for calculating a physically meaningful measure of graph similarity would have immediate practical use in any real-world situation involving network analysis or pattern recognition. Graph comparison can be divided into two main categories: comparing graphs with known vertex correspondence, such as detecting network security breaches \cite{Deltacon}, and the more general problem of comparing graphs with unknown vertex correspondence, like pattern recognition and comparing the molecular structure of organic compounds \cite{balavz1986metric}. The latter problem is much harder to solve, requiring us to consider the combinatorially many ways of labelling either graph. Without restricting the graph topologies, e.g. by requiring the graphs to be sufficiently similar \cite{umeyama1988eigendecomposition} or by considering only trees \cite{jiang1995alignment,dehmer2006similarity}, there exists no classically efficient algorithm for evaluating graph similarity on graphs with unknown vertex correspondence. Even quantum algorithms are restricted in their capacity to provide a meaningful measure of graph similarity due to the exponential difficulty of the problem \cite{quantcompRossi}.

Our algorithm  makes use of a two-part quantum dynamic process: in the first part, we encode information crucial for the comparison of the two graphs in a single qubit. In the second, we call upon techniques in nonlinear quantum computing \cite{abrams1998nonlinear,childs2016optimal} to achieve the speed-up necessary to extract this information efficiently. We use the number of edges in the maximum common edge subgraph as a measure of graph similarity, where the maximum common edge subgraph is defined as the subgraph common to both graphs having the maximal number of edges. The problem of finding the maximum common edge subgraph was first introduced by Bokhari \cite{Bokhari1981MCES}, who showed that it is at least as difficult as the graph isomorphism problem \cite{bahiense2012maximum}. The number of edges in the maximum common edge subgraph is a meaningful measure of graph similarity, as it provides a numerical value $\mathbf{S}(G_1,G_2)$ for two graphs $G_1, G_2$ which follows the axioms for effective graph similarity measures laid out in \cite{Deltacon}. Classically, one must find the maximum edge subgraph itself to find the number of edges in the maximum edge subgraph, and so it is beyond exponentially difficult in the number of vertices to calculate this similarity measure.

The paper is structured as follows: in Section \ref{sec:sec1}, we give an overview of the maximum edge subgraph, and discuss why the maximum edge overlap is a meaningful measure of graph similarity but extremely difficult to calculate. In Section \ref{sec:quant}, we represent all $n!$ permutations of the graph vertices using an efficient quantum circuit operating under linear quantum mechanics.
In Section \ref{sec:grov}, we show that, in terms of the number of graph vertices, the maximum edge overlap is exponentially difficult to extract from the output of the circuit in Section \ref{sec:quant}. In Section \ref{sec:nonlin}, our main body of work, we show how to adapt the circuit from Section \ref{sec:quant} to encode information about the maximum edge overlap in a single-qubit quantum state. We then describe how nonlinear quantum dynamics can be used to efficiently extract the information about the maximum edge overlap. Finally, in Section Section \ref{sec:SandO}, we present the full algorithm for graph comparison and show that its computational cost in the worst-case scenario is efficient in the number of graph vertices.

\begin{comment}\bigskip

\noindent The paper is structured accordingly:
\begin{description}
\item [Section \ref{sec:sec1}] We give an overview of the maximum edge subgraph, and discuss why the maximum edge overlap is a meaningful measure of graph similarity but extremely difficult to calculate.
\item [Section \ref{sec:quant}] We represent all $n!$ permutations on the graph vertices using an efficient quantum circuit operating under linear quantum mechanics.
\item [Section \ref{sec:grov}] We show that, in terms of the number of graph vertices, the maximum edge overlap is exponentially difficult to extract from the output of the circuit in Section \ref{sec:quant}.
\item [Section \ref{sec:nonlin}] In our main body of work, we show how to adapt the circuit from Section \ref{sec:quant} to encode information about the maximum edge overlap in a single-qubit quantum state. We then describe how nonlinear quantum dynamics can be used to efficiently extract the information about the maximum edge overlap.
\item [Section \ref{sec:SandO}] Finally, we present the full algorithm for graph comparison and show that its computational cost in the worst-case scenario is efficient in the number of graph vertices.
\end{description}
\end{comment}

\section{\label{sec:sec1}Maximum edge overlap}

A graph $G(V,E)$ contains a set $V$ of $n$ vertices and a set $E \subseteq V \times V$ of anywhere from 0 to $n(n-1)/2$ edges. A \textit{graph similarity measure} is a function $\mathbf{S}$ that takes two graphs $G_1$ and $G_2$ and returns a similarity score $\mathbf{S}(G_1,G_2) \in [0,1]$ obeying the axioms from \cite{Deltacon}:
\begin{enumerate}
	\item $\mathbf{S}(G_1,G_1) = 1$ for any graph $G_1$ (the identity property),
    \item $\mathbf{S}(G_1,G_2) = \mathbf{S}(G_2,G_1)$ for any two graphs $G_1$ and $G_2$ (the symmetric property), and
    \item $\mathbf{S}(G_1,G_2) \rightarrow 0$ as $n \rightarrow \infty$ (the zero property), where $G_1$ and $G_2$ are the complete and empty graphs on $n$ vertices, respectively.
\end{enumerate}

There are many different ways of defining an exact measure for graph similarity, such as the graph-edit distance \cite{sanfeliu1983distance} and maximum common subgraph. These measures are costly to compute (indeed, finding the maximum common induced subgraph is an NP--hard problem, and finding the maximum common edge subgraph is NP--complete).

Quantum algorithms are by their nature non-deterministic, and hence any quantum procedure to measure graph similarity could be thought of as inexact from the viewpoint of classical computing. However, our algorithm evaluates a graph similarity measure that is classically exact; the only source of non-determinism arises from the nature of quantum measurement, and so we propose that this is in fact an exact measure of graph similarity to within bounded error.

Our measure of graph similarity uses the \textbf{maximum edge overlap}. Given two graphs $G_1$ and $G_2$, each with $n$ labelled vertices, we define the \textit{edge overlap} to be
\begin{equation}\label{eqn:eo}
	\eo(G_1,G_2)=|G_1 \cap G_2|\, ,
\end{equation}
where the intersection $G_1 \cap G_2$ is the graph that contains all edges common to both $G_1$ and $G_2$. The expression $|G_1 \cap G_2|$ denotes the number of edges in this intersection. The \textit{maximum edge overlap} is
\begin{equation}\label{eqn:meo}
\mathbf{MEO}(G_1,G_2) = \max_{\sigma \in S_n} \eo(G_1,\sigma(G_2))\, ,
\end{equation}
where $S_n$ is the set of permutations on $n$ elements and $\sigma(G_2)$ is the graph $G_2$ with vertices relabelled under the permutation $\sigma \in S_n$. The rest of this paper details a quantum algorithm to find the maximum edge overlap for any given pair of graphs. Tweaking Eq.\cref{eqn:meo} slightly results in a measure for graph similarity, $\mathbf{S}$, which abides by the axioms set out in Section \ref{sec:level1}:
\begin{equation}\label{eqn:sim}
    \mathbf{S}(G_1,G_2) = \frac{\mathbf{MEO}(G_1,G_2)}{\max\{|G_1|,|G_2|\}}\, .
\end{equation}
The edge overlap of $G_1$ and any given permutation $\sigma(G_2)$ of $G_2$ may be efficiently calculated. For example,
\begin{align}
    \label{eqn:h2}	\eo(G_1,\sigma(G_2)) &=
    \sum\limits_{i,j=1}^n \left(A_1\right)_{\sigma(i) \sigma(j)}  \left(A_2\right)_{ij}\, ,
    %\left\lvert \left\lfloor \frac{A_1 + \sigma^T A_2 \sigma}{2} \right\rfloor \right\rvert\, ,
\end{align}
where %$H_1(\sigma, A_1, A_2) = |A_1-\sigma^T A_2 \sigma|$ is the number of edges \textit{not} shared by graphs $G_1$ and $\sigma(G_2)$, and 
$A_i$ is the adjacency matrix for graph $G_i$. This can be calculated in $\mathcal{O}(n^2)$ time on a classical computer. %, the permutation $\sigma$ is represented as its corresponding permutation matrix, $\lfloor A \rfloor$ is the element-wise floor of matrix $A$, and $|A|=\sum_{i\leq j} |a_{ij}|$ is the Hamming norm of a symmetric matrix $A$.
Thus the task of finding $\mathbf{S}(G_1,G_2)$ is now the task of finding
\begin{equation} \label{eqn:shrln}
	\mathbf{MEO}(G_1,G_2) = \max_{\sigma \in S_n} \sum\limits_{i,j=1}^n \left(A_1\right)_{\sigma(i) \sigma(j)}  \left(A_2\right)_{ij}\, . %\left\lvert \left\lfloor \frac{A_1 + \sigma^T A_2 \sigma}{2} \right\rfloor \right\rvert\, .
\end{equation}

Consider the black and grey nine-vertex graphs $G_1$ and $G_2$ shown overlaid in Fig. \cref{fig:testgraph}. We can see intuitively that $G_1$ and $\sigma(G_2)$ will share nine edges for any permutation $\sigma$ that maximises $\eo(G_1,\sigma(G_2))$, and so we know that $\mathbf{MEO}(G_1,G_2)=9$.

\begin{figure}[H]
	\centering
    \includegraphics[scale=0.5]{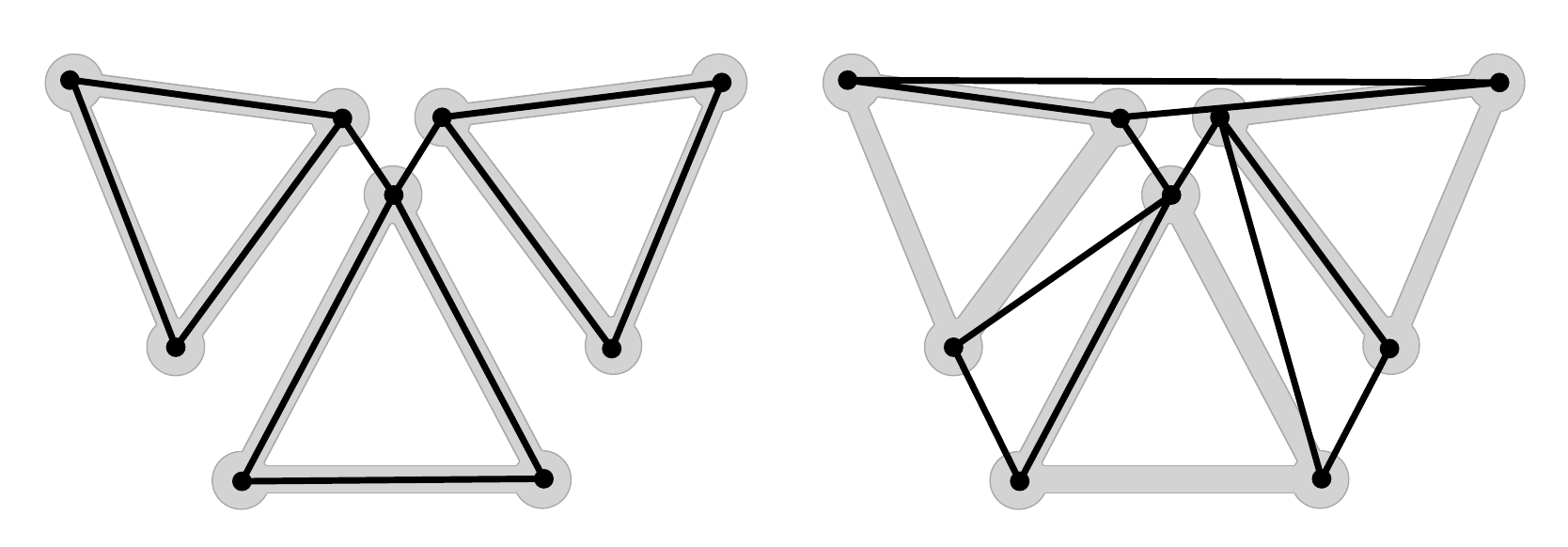}
    \caption{Edge matchings between $G_1$ (grey) and two different vertex permutations of $G_2$ (black). Left: the first permutation of $G_2$ gives the maximum edge overlap, 9. Right: a non-optimal permutation of $G_2$ gives an edge overlap of 3.}
    \label{fig:testgraph}
\end{figure}

To compute $\mathbf{MEO}(G_1,G_2)$ via Eq.\cref{eqn:shrln} requires generating and storing a list of the 9! permutations of the vertices of $G_2$. Each permutation $\sigma$ is then calculated via Eq.\cref{eqn:h2}, and we sort through the $9!$ resulting edge overlaps $\{\eo(G_1,\sigma(G_2)) \, | \, \sigma \in S_n\}$ to find the maximum.

\begin{figure}[H]
    \centering
    \includegraphics[width=0.45\columnwidth,valign=t]{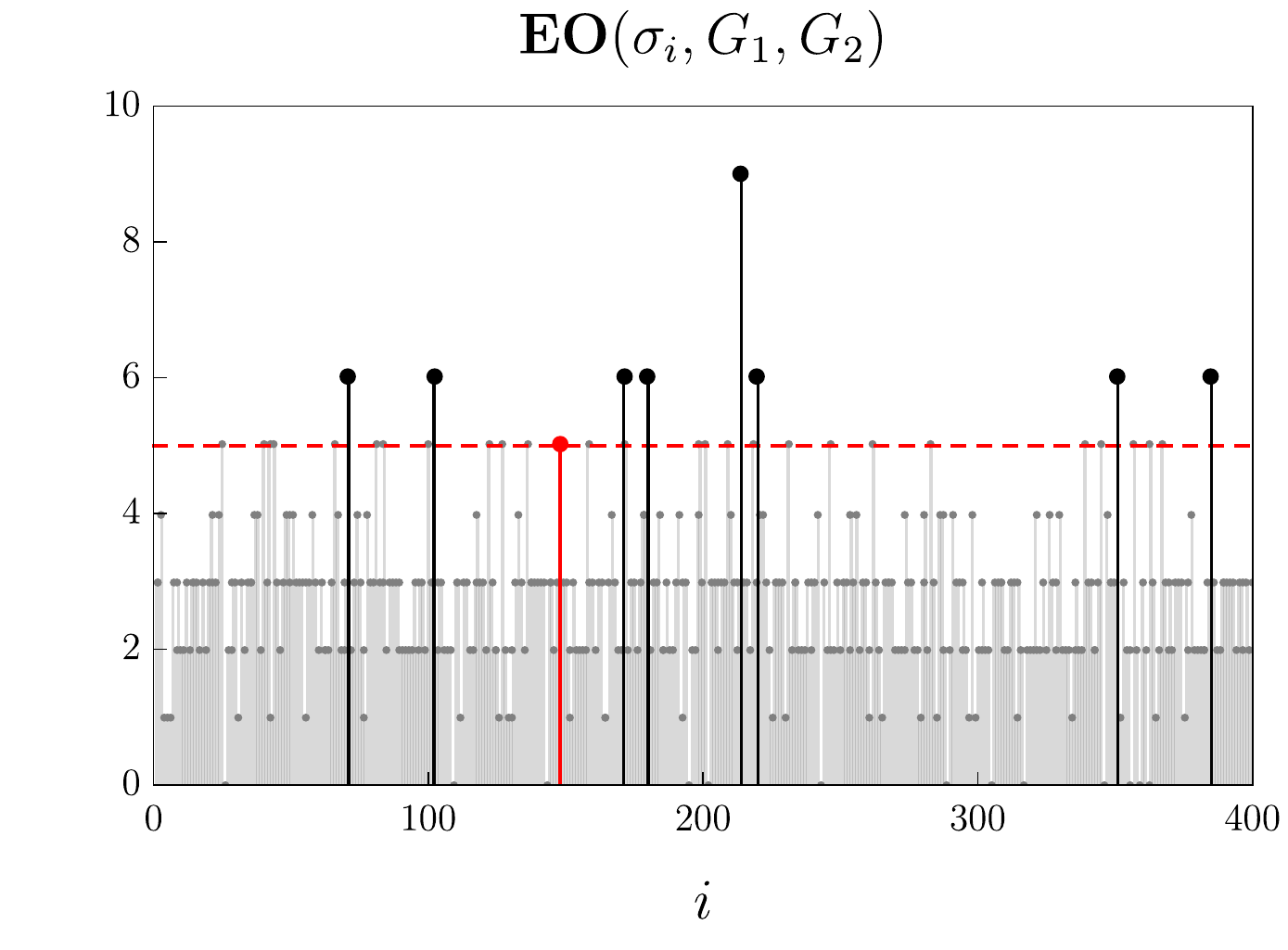}
    \qquad
    \includegraphics[width=0.43\columnwidth,valign=t]{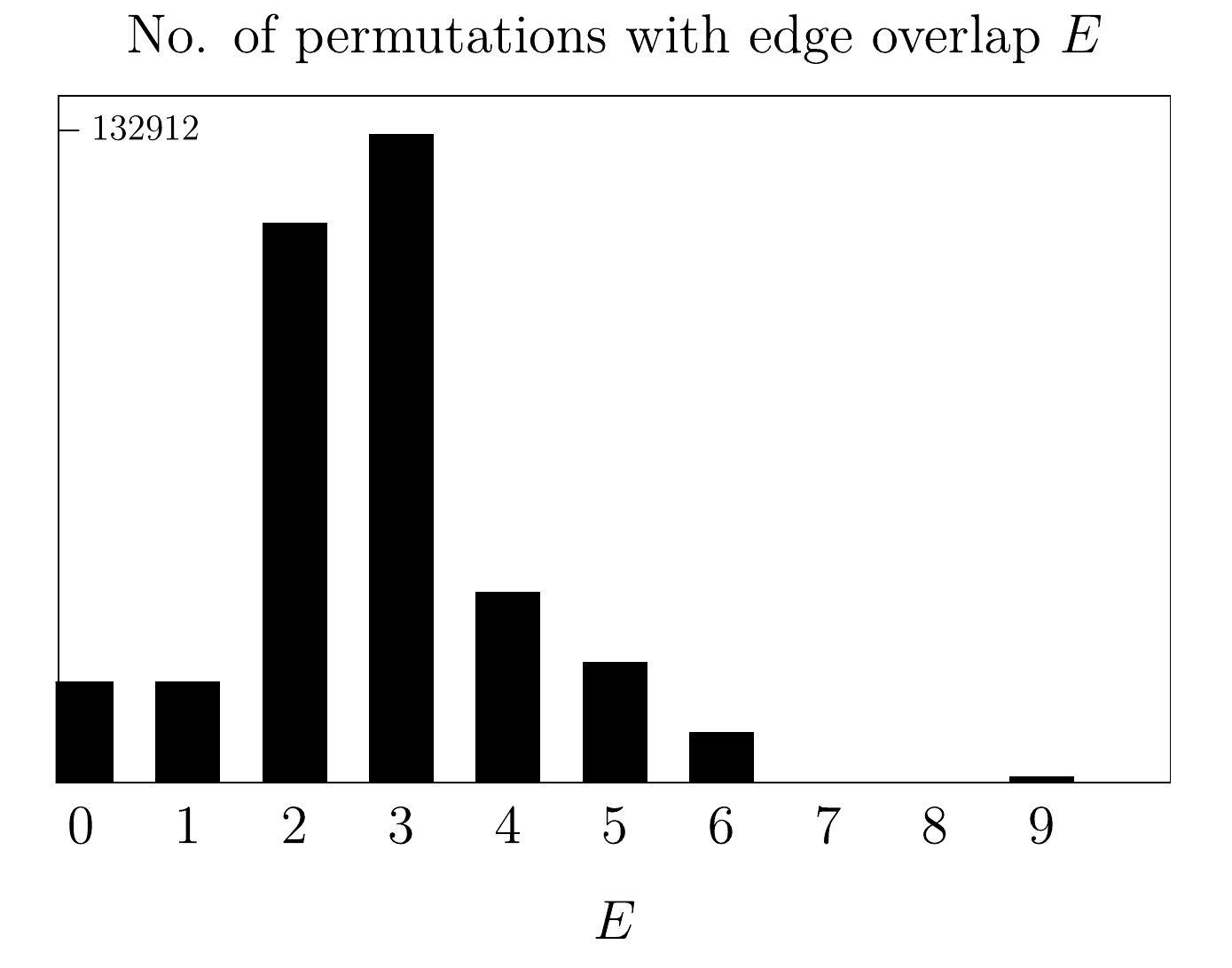}
    \caption{Left: the edge overlaps of 400 permutations of the pair of graphs in Fig. \cref{fig:testgraph}, demonstrating how few exceed an example threshold overlap $E=5$. Right: distribution of edge overlaps for the pair of graphs shown in Fig. \cref{fig:testgraph}.}
    \label{fig:evsbij}
\end{figure}

Fig. \cref{fig:evsbij} shows 400 such edge overlaps, of which only one has the maximum value, 9. Fig. \cref{fig:evsbij} shows a sample of 400 of the edge overlaps, along with a histogram of all 9!. In total, 1296 of the $9!$ permutations give the maximum edge overlap of 9, whose column is barely noticeable in the right figure.

%For any given pair of $n$-vertex graphs with adjacency matrices $A_1$ and $A_2$, the maximum value that $H$ can take for any bijection $\sigma$ is $n(n-1)$. The only situation in which $\min_{\sigma \in S_n} H(\sigma, A_1, A_2)=n(n-1)$ is if $A_1$ and $A_2$ are the complete and empty graphs on $n$ vertices.

%\subsection{Difficulty in classical computation}\label{sec:classdif}

To calculate $\mathbf{MEO}(G_1,G_2)$ via Eq.\cref{eqn:shrln} classically requires evaluating all $n!$ different values for the edge overlap -- one for each $\sigma \in S_n$. A straightforward implementation of Eq.\cref{eqn:h2} has a complexity of $\mathcal{O}(n^2)$, leading to the calculation of maximum edge overlap having complexity of $\mathcal{O}(n! \, n^2)$. The requirement to iterate $n!$ times dominates this expression, and so $\mathbf{MEO}(G_1,G_2)$ cannot be computed efficiently via classical means. To circumvent the need to calculate edge overlaps $n!$ separate times, we make a foray into the realm of quantum computing.

%\textit{NOTE:} A second, less apparent challenge is that the number of permutations inducing the maximum edge overlap is so low, and as a result they can go virtually undetected. The proportion of permutations inducing maximum edge overlap varies case-by-case, but generally behaves as $\mathcal{O}(\frac{\poly(n)}{n!})$. This also prevents the effective application of quantum amplitude amplification, as described at the end of Section \ref{sec:quant}.

\section{A quantum circuit to store permutations for optimisation}\label{sec:quant}

Quantum computing presents us with the powerful tool of quantum superposition, lacking in any classical device. We can bypass the classical hurdle of calculating all $n!$ edge overlaps separately through the generation and manipulation of a quantum superposition of states representing all $n!$ permutations, $\sum_{\sigma \in S_n} \ket{\sigma}/\sqrt{n!}$, where each $\ket{\sigma}$ is a quantum state that encodes relevant information about the permutation $\sigma \in S_n$. We can store this superposition state using $\log(n!)=\mathcal{O}(n \log (n))$ qubits.

It is standard practice in quantum optimisation algorithms, given a superposition of quantum states as input, to specify a threshold and mark all terms that exceed the threshold \cite{durr1996quantum,baritompa2005grover}. In our case, the threshold is a user-controlled value $E$ with $0 \leq E \leq \min\{|G_1|,|G_2|\}$. We mark the terms on an ancillary single-qubit register which is conditionally mapped from $\ket{0}$ to $\ket{1}$ if the term containing $\ket{\sigma}$ has $\eo(\sigma)>E$, where $\eo(\sigma)$ is an abbreviation for $\eo(G_1,\sigma(G_2))$.
The procedure to implement this marking uses three quantum registers containing at most $2(n+1)\log n$ qubits and $\mathcal{O}(n^2 \log^2(n))$ fundamental quantum gates, and produces the mapping 
\begin{equation}\label{eqn:sec2}
\begin{aligned}
	\ket{0}^{\otimes \mathcal{O}(\log(n))} \otimes \ket{1}\ket{2}\dots \ket{n} \otimes \ket{0}^{\otimes \Theta(n \log(n))} \ket{0} \longmapsto \ket{0}^{\otimes \mathcal{O}(\log(n))} \ket{1}\ket{2}\dots\ket{n} \ur \left(\frac{1}{\sqrt{n!}}\sum_{\sigma \in S_n} \ket{\sigma}_{\mathrm{radix}} \ket{0}\right)\, ,
\end{aligned}
\end{equation}
where\begin{equation}\label{eqn:desire}
\ur\left(\frac{1}{\sqrt{n!}}\sum_{\sigma \in S_n} \ket{\sigma}_{\mathrm{radix}} \ket{0}\right) \coloneqq \frac{1}{\sqrt{n!}}\left(\sum_{\eo(\sigma)\leq E} \ket{\sigma}_{\mathrm{radix}} \ket{0} + \sum_{\eo(\sigma) > E} \ket{\sigma}_{\mathrm{radix}}\ket{1}\right)\, .
\end{equation}
Reverse computation on the first two registers, performed as part of the process in Eq.\cref{eqn:sec2}, ensures that the mapping only alters the third and single-qubit ancilla registers. These registers contain the marked superposition state we desire, shown in Eq.\cref{eqn:desire}. The state $\ket{\sigma}_{\mathrm{radix}}$ is of a particular form within which we encode the permutation $\sigma \in S_n$ on a quantum register, distinct from $\ket{\sigma}$. Read left-to-right, the state in Eq.\cref{eqn:sec2} is comprised of $\mathcal{O}(\log(n))$-, $\mathcal{O}(n \log(n))$-, $\Theta(n\log(n))$- and single-qubit systems respectively, requiring $\mathcal{O}(n\log (n))$ qubits in total across the four registers. Fig. \ref{fig:threshold} visually breaks down the different stages of the mapping in Eq.\cref{eqn:sec2} into their separate components.

\begin{figure}[H]
	\centering
    \hspace{-0.05\columnwidth}
    \includegraphics[width=0.8\columnwidth]{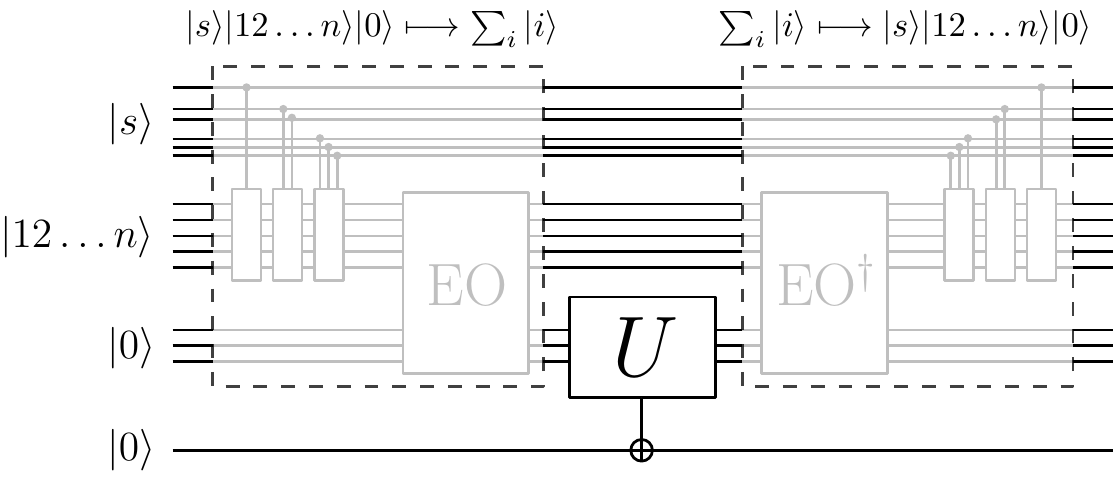}
    \caption{The circuit we construct in Section \ref{sec:quant} to mark permutations with edge overlap greater than some threshold $E$.}
    \label{fig:threshold}
\end{figure}

Note that the end result of Eq.\cref{eqn:sec2} is to implement a controlled-NOT gate on the ancillary qubit conditioned by the parity of $(E-\eo(\sigma_i))$ for each term $\ket{\sigma_i}_{\mathrm{radix}}$ in the third register. If the reader is willing to assume the existence and efficiency of such a quantum process, they may now proceed directly to Section \ref{sec:grov}. Due to the nontrivial task of representing the elements of $S_n$ on a quantum register, we believe it is necessary to provide the reader with our method of implementing Eq.\cref{eqn:sec2}. 
%We need a reversible process (this condition manifests in the presence of the $G^{-1}$ gate in Fig. \ref{fig:hadamardtest}).
What follows is a procedure similar to the algorithm for permutation generation demonstrated in \cite{Abrams1997perms}.

\subsection{Sub-algorithm for generating a superposition of all permutations}

Every permutation $\sigma \in S_n$ can be represented in what is known as the mixed radix system \cite{sedgewick1977perms} by an array $C[1], C[2], C[3], \dots, C[n]$, where the value $C[i]$ is defined to be the number of elements smaller than $i$ which appear to the left of $i$ in the list
\begin{equation}\label{eqn:bijrep}
	\{\sigma(1), \sigma(2), \sigma(3), \dots, \sigma(n)\}.
\end{equation}

As an example, take the cyclic permutation $\sigma = (1\,2\,4\,6\,5\,3) \in S_6$ which maps the list $(1,2,3,4,5,6)$ to $(2,4,1,6,3,5)$. Thus $C[2]=0$, $C[4]=1$, $C[1]=0$, $C[6]=3$, $C[3]=2$, $C[5]=4$, and so $\ket{\sigma}_{\mathrm{radix}}=\ket{0}\ket{0}\ket{2}\ket{1}\ket{4}\ket{3}=\ket{002143}$. Notice that the radix numbering systems lends itself to a natural ordering of permutations according to the increasing decimal value of their radix numbers, and it is according to this numbering scheme that we define $\sigma_i$ to be the $i$th permutation (i.e. $\ket{\sigma_{1}}_{\mathrm{radix}} = \ket{000000}, \ket{\sigma_{2}}_{\mathrm{radix}}=\ket{000001}$, and so on).

The representation $\ket{\sigma_i}_{\mathrm{radix}}$, while helpful in enumerating permutations, does not help us with subsequent calculation of $\eo(G_1,\sigma_i(G_2))$ via the formulation in Eq.\cref{eqn:h2}. To that end, we need access to the values $\{\sigma_i(1), \sigma_i(2), \dots, \sigma_i(n)\}$. We achieve this through introduction of a second register containing $n$ subsystems of $\log(n)$ qubits. The second register is shown in Fig. \ref{fig:threshold}, and in more detail in Fig. \ref{fig:permucirc}. Across all of its subsystems, it initially holds the product state $\ket{1}\ket{2}\dots\ket{n}$ before we map it to $\ket{\sigma_i(1)}\ket{\sigma_i(2)}\dots \ket{\sigma_i(n)}$ via careful conditioning on the permutation $\ket{\sigma_i}_{\mathrm{radix}}$ stored in the first register. The sub-algorithm we present in this section achieves the mapping
\begin{equation}\label{eqn:desired}
\ket{1} \ket{2} \dots \ket{n} \otimes \ket{0}^{\otimes \Theta(n \log(n))} \longmapsto \frac{1}{\sqrt{n!}} \sum_{i=1}^{n!} \big(\ket{\sigma_i(1)}\ket{\sigma_i(2)}\,\dots\,\ket{\sigma_i(n)}\big) \otimes  \ket{\sigma_i}_{\mathrm{radix}}\, .
\end{equation}
With this step complete, we can use the second register to calculate all of the edge overlaps on a third register. Before proceeding, however, we prove Eq.\cref{eqn:desired} and show that it takes time $\mathcal{O}(n^2\log^2(n))$ in terms of fundamental quantum gates. Perform the following steps:

\begin{enumerate}
	\item Generate each of the $n$ states $|0\rangle$, $\frac{|0\rangle+|1\rangle}{\sqrt{2}}$, $\dots$, $\frac{|0\rangle+|1\rangle+\dots+|n-1\rangle}{\sqrt{n}}$ on the first register, which consists of $n$ systems of appropriate size with each system initiated in its respective $\ket{0}$ state. The systems storing $\ket{0}$ and $\frac{1}{\sqrt{2}}(\ket{0}+\ket{1})$ each require one qubit, storing $\frac{1}{\sqrt{3}}(\ket{0}+\ket{1}+\ket{2})$ requires $\lceil\log(3)\rceil=2$ qubits, and so on. The final system requires $\lceil\log(n)\rceil$ qubits, and thus the total number of qubits in the first register is $\Theta(n \log(n))$. Details on how to generate each of these states from the state $\ket{0}^{\otimes \Theta(n \log(n))}$ (a process denoted by $G$ in Fig. \ref{fig:threshold}) are explained in Section \ref{sec:subalgtime}.
Considered as one whole register (register \raisebox{.5pt}{\textcircled{\raisebox{-.9pt} {1}}} in Fig. \cref{fig:permucirc}), we now have the state
 \begin{align*}
\ket{s} \coloneqq \frac{1}{\sqrt{n!}} \sum_{\alpha_2=0}^1\dots\sum_{\alpha_n=0}^{n-1}\ket{0}\ket{\alpha_2}\dots\ket{\alpha_n} = \frac{1}{\sqrt{n!}} \sum_{C[2]=0}^1\dots\sum_{C[n]=0}^{n-1}\ket{0}\ket{C[2]}\dots\ket{C[n]} = \frac{1}{\sqrt{n!}} \sum_{i=1}^{n!} \ket{\sigma_i}_{\mathrm{radix}}\, .
\end{align*}
%which is in effect a superposition of all the arrays $C[2], C[3], \dots, C[n]$.

\item Initiate the second register, which consists of $n$ sub-systems, in the state $\ket{1}\ket{2}\dots\ket{n}$. This represents the elements of the set $\{1,2,\dots,n\}$, and shown on register \raisebox{.5pt}{\textcircled{\raisebox{-.9pt} {2}}} in Fig. \cref{fig:permucirc}. These systems each need $\log(n)$ qubits as they will be permuted amongst each other, rendering the size of the entire second register as $n \log (n)$ qubits. Combining the first and second register now, we have
 \begin{equation}
 \begin{aligned}
\frac{1}{\sqrt{n!}} \sum_{i=1}^{n!} \ket{1} \ket{2} \dots \ket{n} \otimes  \ket{\sigma_i}_{\mathrm{radix}}\, .
\end{aligned}
\end{equation}

\item According to Hall's strategy \cite{sedgewick1977perms}, we perform a sequence of unitary operations $P_2,P_3,\dots,P_n$ in order on the two-register state, where $P_i$'s action on the first register is limited to the $i$th subsystem and is defined to be
\begin{align}\label{eqn:pdef}
P_i &=\prod_{l=1}^{i-1} \left(S_{l,l+1}\bigotimes \sum_{j=0}^{l-1}|j\rangle\langle j|+\mathbb{I}\, \bigotimes \sum_{j=l}^{2^{\lceil \log_2(i)\rceil}-1}|j\rangle\langle j|\right)\\ \label{eqn:pdef2}
&=\sum_{j=0}^{i-2}\left( \prod_{l=j+1}^{i-1} S_{l,l+1} \bigotimes |j\rangle\langle j|\right)+\sum_{j = i-1}^{2^{\lceil \log_2(i)\rceil}-1}\left(\mathbb{I} \, \bigotimes |j\rangle\langle j|\right)\, .
\end{align}
Here $S_{l,l+1}$ is the SWAP operation that permutes the $l$th and  $(l+1)$th systems of the second register. Note that $\prod_{l=j+1}^{i-1} S_{l,l+1}=S_{j+1,j+2}S_{j+2,j+3}\dots S_{i-1,i}$ is a series of SWAP operations that shifts the $i$th system of the second register down to the $(j+1)$th system and increases the indices of all intermediary systems by one. Once this is done, we finally obtain the superposition state
\begin{equation}\label{eqn:radixperms}
 \begin{aligned}
\frac{1}{\sqrt{n!}}\sum_{i=1}^{n!} \ket{\sigma_i(1)} \ket{\sigma_i(2)} \dots \ket{\sigma_i(n)} \otimes \ket{\sigma_i}_{\mathrm{radix}}\, ,
\end{aligned}
\end{equation}
which is the same as what was required in Eq.\cref{eqn:desired}.
\end{enumerate}

\begin{center}
\begin{figure}[H]
	\centering
	\includegraphics[scale=1.0]
	{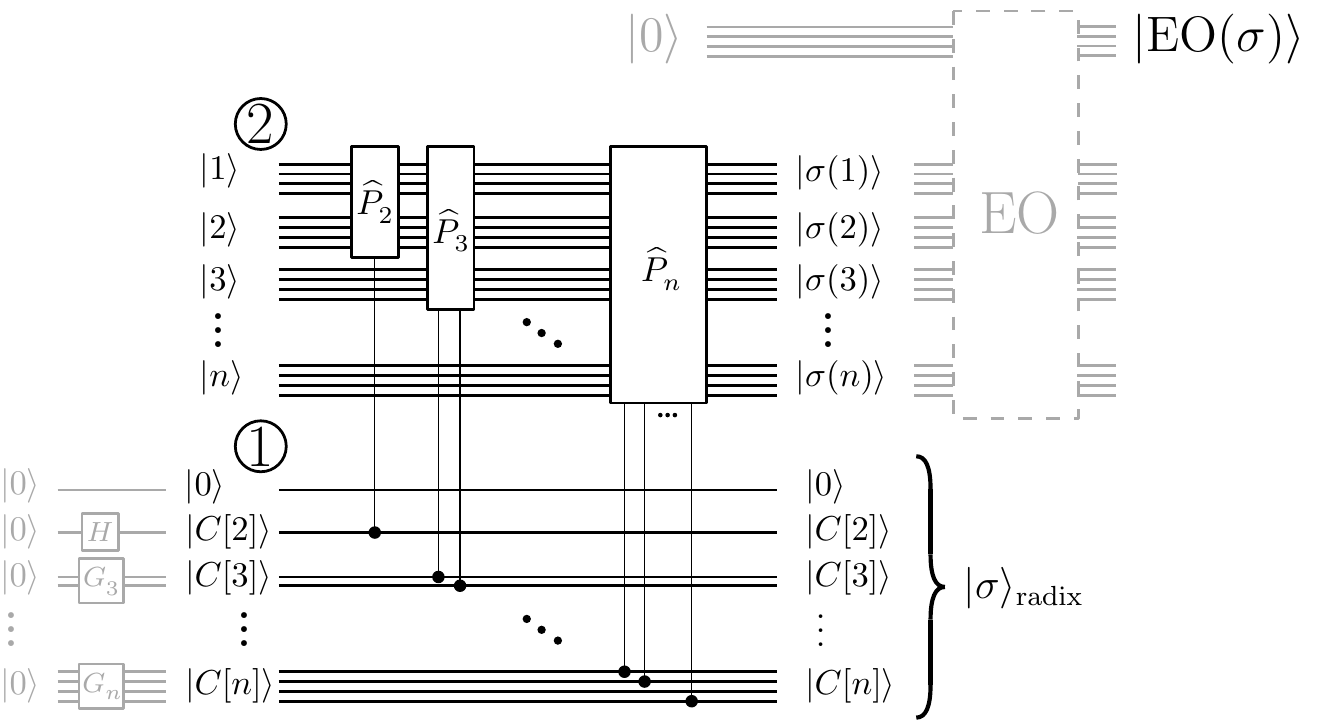}
	\caption{The black circuit performs the mapping $\ket{1}\ket{2}\dots\ket{n} \ket{\sigma}_{\mathrm{radix}} \mapsto \ket{\sigma(1)}\ket{\sigma(2)}\dots\ket{\sigma(n)} \ket{\sigma}_{\mathrm{radix}}$. The second register \raisebox{.5pt}{\textcircled{\raisebox{-.9pt} {2}}} stores the elements of $\{1,2,\dots,n\}$, and is conditionally permuted by register \raisebox{.5pt}{\textcircled{\raisebox{-.9pt} {1}}}. Register \raisebox{.5pt}{\textcircled{\raisebox{-.9pt} {1}}} can be preceded by the gate $G$, as shown, which replaces $\ket{C[2]}$ with $(\ket{0}+\ket{1})/2$, $\ket{C[3]}$ with $(\ket{0}+\ket{1}+\ket{2})/\sqrt{3}$ and so on in order to generate the state $\ket{s}=\sum_{\sigma \in S_n} \ket{\sigma}_{\mathrm{radix}}/\sqrt{n!}$. If this circuitry is included, then the black circuit will output $\sum_{\sigma \in S_n}\ket{\sigma(1)}\ket{\sigma(2)}\dots\ket{\sigma(n)} \ket{\sigma}_{\mathrm{radix}} $. Here, the controlled-$\hat{P}_i$ $(2 \leq i \leq n)$ operation corresponds to $P_i$. Overall, this circuit prepares the second register in an appropriate input format for a quantum gate implementing the calculation of edge overlap.} 
	\label{fig:permucirc} 
\end{figure}
\end{center}

It is perhaps easiest to understand Eq.\cref{eqn:pdef} and Eq.\cref{eqn:pdef2} by writing the first few operations out:
\begin{align*}
	P_2 &= S_{1,2} \otimes \ket{0}\bra{0} + \mathbb{I} \otimes \ket{1}\bra{1}\, , \\
    P_3 &= S_{1,2} S_{2,3} \otimes \ket{0}\bra{0} + S_{2,3}\otimes \ket{1}\bra{1} +  \mathbb{I} \otimes \big( \ket{2}\bra{2} + \ket{3}\bra{3} \big)\, , \\
    P_4 &= S_{1,2}S_{2,3}S_{3,4} \otimes \ket{0}\bra{0} + S_{2,3}S_{3,4}\otimes \ket{1}\bra{1} + S_{3,4} \otimes \ket{2}\bra{2} + \mathbb{I} \otimes \ket{3}\bra{3} \, ,
\end{align*}
and so on. The sequence $P_2, P_3, \dots, P_n$ implements the permutation encoded in the first register upon the representation of $\{1,2,\dots,n\}$ stored in the second register. For example, let $n=6$ and take one of the superposition terms, $\ket{123456} \otimes \ket{\sigma}_{\mathrm{radix}}$ where $\ket{\sigma}_{\mathrm{radix}}=\ket{002143}$. Then,
 \begin{eqnarray*}
\ket{123456} \ket{002143}  && \xrightarrow{P_2} \ket{213456}\ket{002143} \xrightarrow{P_3} \ket{213456} \ket{002143} \xrightarrow{P_4} \ket{241356} \ket{002143}  \\
&&\xrightarrow{P_5} \ket{241356} \ket{002143} \xrightarrow{P_6} \ket{241635} \ket{002143}\, .
\end{eqnarray*}

\subsubsection{Time complexity of sub-algorithm}\label{sec:subalgtime}

The time complexity of this algorithm is evaluated by estimating the number of basic one-qubit and two-qubit gates. For each integer $k$ in the interval $[2,n]$, the gate $G$ in Step 1 generates the state $(|0\rangle+|1\rangle+\dots+|k-1\rangle)/\sqrt{k}$ in a subsystem of $\lceil \log_2(k) \rceil $ qubits. If $k$ is a power of $2$, this state can be easily generated via $H^{\otimes \lceil \log_2(k)\rceil}$, where $H$ is the Hadamard gate on one qubit. If $k$ has any other value, the state can be generated from the corresponding $\ket{0}$ state of the same dimension via generalized amplitude amplification (Theorem 4 of \cite{brassard2000amplampl}). On each subsystem of $\lceil \log_2(k) \rceil$ qubits, we employ the unitary operation $G_k=-H^{\otimes \lceil \log_2(k) \rceil}S_0(\phi)H^{\otimes \lceil \log_2(k) \rceil}S_\chi(\varphi) H^{\otimes \lceil \log_2(k)\rceil}$, where the conditional phase shifts $S_0$ and $S_{\chi}$ are defined by
\begin{equation}\label{eqn:s0}
S_0(\phi)\ket{j} = \begin{cases} \ee^{\ii \phi}\ket{j} & j = 0\\ \ket{j} & j \neq 0\end{cases}
\end{equation}
and
\begin{equation}\label{eqn:sx}
S_\chi(\varphi)\ket{j} = \begin{cases} \ee^{\ii \varphi}\ket{j} & 0\leq j < k\\ \ket{j} & k \leq j < n\, .\end{cases}
\end{equation}
%$S_\chi(\varphi)$ denotes the unitary operation that multiplies the amplitudes by a factor of $e^{i\varphi}$ if and only if the state is one of the states $\ket{0},\ket{1},\dots,\ket{k-1}$, $S_0(\phi)$ denotes the unitary operation that multiplies the amplitudes by a factor of $e^{i\phi}$ if and only if the state is $\ket{0}$.
Both $S_0$ and $S_{\chi}$ can be implemented efficiently, using $\mathcal{O} (\lceil \log(k) \rceil)$ gates, via a technique called phase kick-back \cite{cleve1998quantum,lee2015generalised}. To ensure $G_k$ performs the exact required rotation, a condition arises restricting the values of $\varphi$ and $\phi$ \cite{brassard2000amplampl}, from which valid numerical values for these parameters can be determined. Because the proportion of desirable amplitudes $k/n$ is greater than 1/2, no prior Grover iterations (as discussed in \cite{brassard2000amplampl}) are required and we simply have $G_k\ket{0}=(|0\rangle+|1\rangle+\dots+|k-1\rangle)/\sqrt{k}$. Since both $S_\chi(\varphi)$ and $S_0(\phi)$  take time $\mathcal{O}(\lceil \log(k) \rceil)$ \cite{Barenco1995elementary}, it takes time $\mathcal{O}(\log(k))$ to implement $G_k$ and total time $\mathcal{O}(n\log(n))$ to implement all of $G$.

In Step 2, the state $|12\dots n\rangle$ can be generated from  $|00 \dots 0\rangle$ using $\mathcal{O}(n\log (n))$ one-qubit gates. From Eq.~\eqref{eqn:pdef}, we can see that $P_i$ takes $i-1$ sequential controlled-SWAP operations and each SWAP operation uses $\mathcal{O}(\log(n)+\log(i))=O(\log(n))$ fundamental gates \cite{Barenco1995elementary,Berry2018Comparator}, so Step 3 takes time 
$\mathcal{O}\left(\log(n)(1+2+\dots+n-1)\right)=\mathcal{O}(n^2\log(n))$. Combining all three steps, preparing a superposition of all permutations takes time $\mathcal{O}(n\log (n)+n\log (n)+n^2\log(n))=\mathcal{O}(n^2\log(n))$, from which we can see that Step 3 dominates the runtime of this sub-algorithm. As discussed, both registers require at most $n \log(n)$ qubits.

Note that the complexity being greater than $\mathcal{O}(n \log(n))$ shows how this sub-algorithm is necessary -- constructing a superposition of all permutations takes longer than it does to construct the much simpler $\sum_{i=1}^{n!}\ket{i}$, which would only take $\mathcal{O}(\log(n!)) = \mathcal{O}(n \log (n))$ time. Intuitively, this is because the latter requires further computation to extract usable information. To proceed any further, we need to be able to calculate the edge overlap for each permutation; merely having an index $i$ for each permutation gets us no closer to our goal.
% * <samuel.marsh@research.uwa.edu.au> 2018-04-24T04:29:17.797Z:
% 
% > This is because the latter does note encode enough information for the subsequent easy calculation of $\eo$.
% Need to rephrase in terms of oracles I think

\subsection{Completing the map}

Recall from the start of Section \ref{sec:quant} that we set a threshold edge overlap, $E$. Proceeding from Eq.\cref{eqn:desired}, we introduce a third register of $\log(\min\{|G_1|,|G_2|\})=\mathcal{O}(\log(n))$ qubits initially in the state $\ket{0}^{\otimes \mathcal{O}(\log(n))}$. Our next step towards achieving the map in Eq.\cref{eqn:sec2} is to induce
\begin{equation}\label{eqn:hmaps}
		\frac{1}{\sqrt{n!}}\sum_{i=1}^{n!} \ket{0}^{\otimes \mathcal{O}(\log(n))} \ket{\sigma_i}  \ket{\sigma_i}_{\mathrm{radix}}  \longmapsto \frac{1}{\sqrt{n!}}\sum_{i=1}^{n!} \ket{\eo(\sigma_i)} \ket{\sigma_i}\ket{\sigma_i}_{\mathrm{radix}}\, ,
\end{equation}
where $\ket{\sigma_i} = \ket{\sigma_i(1)\,\sigma_i(2)\,\dots\,\sigma_i(n)}$ and $\ket{\eo(\sigma_i)}$ contains the value $\eo(\sigma_i)$. Note that, given the existence of an efficient mapping $\eo$ on the third and second registers,
\begin{equation}\label{eqn:quantH}
	\eo:\ket{0}^{\otimes \mathcal{O}(\log(n))} \ket{\sigma_i} \longmapsto  \ket{\eo(\sigma_i)} \ket{\sigma_i}\, ,
\end{equation}
Eq.\cref{eqn:hmaps} unfolds routinely. Indeed, the mapping $\eo$ can be implemented efficiently in the quantum regime, as its classical counterpart is efficient. Assuming oracular access to the graphs' adjacency matrices, a straightforward quantum circuit to carry out EO via Eq.\cref{eqn:h2} will use $\mathcal{O}(n^2 \log n)$ fundamental quantum gates. This is done by adding the value of $\left(A_1\right)_{\sigma(i) \sigma(j)}  \left(A_2\right)_{ij}$ to the register storing the edge overlap, for each $i, j = 1 \ldots n$. Each addition can be performed using $\mathcal{O}(\log (n))$ gates \cite{vedral1996}, and there are $n^2$ additions in total. Thus the runtime of the quantum circuit implementing Eq.\cref{eqn:quantH} will scale as $\mathcal{O}(n^2 \log (n))$.

We can use a fourth register consisting of a single ancillary qubit to mark permutations whose edge overlap exceeds the threshold $E$. As shown in Fig. \ref{fig:threshold}, we achieve this with the comparator $(>E)$ gate, which performs
\begin{align} \label{eqn:supmarked}
	\frac{1}{\sqrt{n!}} \sum_{i=1}^{n!} \ket{\eo(\sigma_i)} \ket{\sigma_i} \ket{\sigma_i}_{\mathrm{radix}} \ket{0} \longmapsto \frac{1}{\sqrt{n!}} \bigg(\sum_{\eo(\sigma_i) \leq E} \ket{\eo(\sigma_i)} \ket{\sigma_i} \ket{\sigma_i}_{\mathrm{radix}}\ket{0} + \sum_{\eo(\sigma_i)>E} \ket{\eo(\sigma_i)} \ket{\sigma_i} \ket{\sigma_i}_{\mathrm{radix}}\ket{1} \bigg)\, .
\end{align}
Implementing the comparator gate takes $\mathcal{O}(\log(n))$ fundamental quantum operations and $\mathcal{O}(\log(n))$ additional ancilla qubits, using the method outlined in the supplementary materials of \cite{Berry2018Comparator}. Note that we are nearing the desired output as written in Eq.\cref{eqn:sec2}; all that remains is to reverse the operations done on the second and third registers. This leaves the entire four-register system in the state
\begin{align}
&\ket{0}^{\otimes \mathcal{O}(\log(n))} \otimes \ket{1} \ket{2} \dots \ket{n} \otimes \frac{1}{\sqrt{n!}}\bigg(\sum_{\eo(\sigma)\leq E} \ket{\sigma}_{\mathrm{radix}} \ket{0} + \sum_{\eo(\sigma) > E} \ket{\sigma}_{\mathrm{radix}}\ket{1}\bigg) \notag\\ = &\ket{0}^{\otimes \mathcal{O}(\log(n))} \otimes \ket{1} \ket{2} \dots \ket{n} \otimes \ur \big(\ket{s}\ket{0}\big)\, ,
\end{align}
as required. The reversal of computation merely doubles the fundamental quantum operation count, so the cost of the entire algorithm up to this point is still $\mathcal{O}(n^2 \log^2(n))$. The first register contains $\Theta(n\log(n))$ qubits, the second $n\lceil \log(n) \rceil$, and the third $\lceil \log(n(n-1)/2) \rceil$, and so we require less than $2(n+1)\log(n)=\mathcal{O}(n\log(n))$ qubits. This completes the map described in Eq.\cref{eqn:sec2}.

We have produced a state that encodes all the information derived classically in Fig. \cref{fig:evsbij}. The cost of preparing this state is $\mathcal{O}(n^2 \log n)$ fundamental quantum gates, dominated by preparing the radix-form superposition over all permutations. We have thus bypassed the first of the two hurdles mentioned in the classical attempt at the problem in Section \ref{sec:level1} -- the problem of calculating all $n!$ edge overlaps efficiently.

%-----------------------------------------

\section{Failure of linear quantum optimisation protocols to find the maximum edge overlap}\label{sec:grov}

To build an optimisation algorithm to find $\mathbf{MEO}(G_1,G_2)$, we need to be able to efficiently determine if any permutations have edge overlap less than $E$. The quantum state produced by the map in Eq.\cref{eqn:sec2} contains the information we need -- some collection of its terms would reveal the maximum edge overlap upon successful measurement -- but how do we extract them?  In the regime of linear quantum algorithms, standard intuition would have us use D{\"u}rr and H{\o}yer's algorithm \cite{durr1996quantum,baritompa2005grover} for quantum maximum searching, detailed in Algorithm \ref{alg:durr}. 

\begin{algorithm}[H]
\caption{Quantum Maximum Searching Algorithm (insufficient)}
\label{alg:durr}
  \begin{algorithmic}
    \INPUT{$G_1, G_2$}
    \State generate $y \in \{1,\dots,n!\}$ uniformly, and set threshold $E = \eo(\sigma_y)$.
    \State set $p=1$ and $\lambda \in (1, 4/3)$.
    \For {$i=1,2,\dots$ until running time exceeds $22.5\sqrt{n!}+1.4\log^2(n!)$}
    	\State{initialise the three quantum registers in the state $\ket{0}^{\otimes \mathcal{O}(\log(n))}\otimes  \ket{1}\ket{2} \dots \ket{n} \ket{0}^{\Theta(n \log(n))} \ket{0}$.}
        \State{mark permutations $\ket{\sigma_i}$ for which $\eo(\sigma_i)>E$ via the circuit described in Section \ref{sec:quant}.}
        \State{perform a number of Grover rotations on the first register chosen uniformly in the interval $\{0, \dots, \lceil p-1\rceil\}$.}
        \State{observe the first register and measure $\ket{\sigma_k}$. If $\eo(\sigma_k) > E$, set $E = \eo(\sigma_k)$ and $p = 1$. Otherwise, set $p = \lambda p$.}
    \EndFor
    \State \Return $E$
  	\OUTPUT{$E = \mathbf{MEO}(G_1,G_2)$ with probability $>1/2$.}
  \end{algorithmic}
\end{algorithm}

One can repeat Algorithm \ref{alg:durr} to improve the probability of success. \begin{comment} The following specific time costs are incurred:
\begin{enumerate}
	\skipitems{2}
    \item \begin{enumerate}
    		\item Preparing the superposition $\ket{0}\ket{12\dots n}\ket{0 0 \dots 0}\ket{0}$ takes time $\mathcal{O}(n \log(n))$. 
            \item As shown in \cref{sec:quant}, applying the marking procedure $\ur$ takes time $\mathcal{O}(n^2 \log(n))$.
            \item Each Grover rotation makes use of $\ur$ and the process of Step 3.(b) twice. The precise number of Grover rotations used is random, but the total complexity of the overall algorithm is such that this step never exceeds $\mathcal{O}(n^2\log(n))$ time. 
    		\end{enumerate}
\end{enumerate}

\end{comment}
Note that the maximum runtime of $22.5 \sqrt{n!} + 1.4 \log^2(n!)$ is infeasible, as it is beyond exponential in $n$. Also, as the threshold $E$ narrows down on $\mathbf{MEO}(G_1,G_2)$, it will become exceedingly unlikely that we will ever measure a $\ket{\sigma_k}$ such that $\eo(\sigma_k)>E$: as we are presuming $G_1$ and $G_2$ to be general graphs, we must be prepared to accept that the amplitude of the marked component of the output in Eq.\cref{eqn:sec2} becomes vanishingly small, on the order of $\mathcal{O}(1/\sqrt{n!})$.

(A quick justification of this logic: we cannot assume the number of permutations giving maximum edge overlap to be anything more than constant. For example, the worst-case scenario occurs when comparing two line graphs of $n$ vertices. The graphs are isomorphic, but there are only \textit{two} optimal permutations and so the amplitude of marked terms will never exceed $\sqrt{2/n!}$. Thus, there is no obvious structure or predictability in the problem of calculating $\mathbf{MEO}(G_1,G_2)$ which can be exploited: we must always assume that there is a constant number of MEO permutations amongst all $n!$ other permutations.)

We conclude that the standard quantum computing model is not powerful enough to significantly wrest the difficulty of finding $\mathbf{MEO}(G_1,G_2)$ from inefficiency to that of an algorithm performing in polynomial time, where $G_1$ and $G_2$ are general graphs.

\section{Nonlinear quantum computing}\label{sec:nonlin}

There is, however, a different quantum algorithmic route to finding the maximum edge overlap in $\poly(n)$ time, combining the pioneering work of Abrams and Lloyd \cite{abrams1998nonlinear} and a recently proposed algorithm by Childs and Young \cite{childs2016optimal}, both of which exploit the properties of nonlinear quantum computing.

The potential of nonlinear quantum dynamics in quantum computing is an emergent field of study. In particular, the recent series of publications \cite{meyer2013nonlinear,meyer2014quantum,kahou2013quantum} used the Gross-Pitaevskii dynamics of interacting Bose Einstein condensates to perform Grover’s search at a runtime which scales as $O(\min\{\sqrt{N/g},\sqrt{N}\})$, where $g$ denotes the nonlinearity strength and $N$ the database size (in our case, $N=n!$). Since then, Childs and Young proposed their nonlinear protocol using the same Gross-Pitaevskii dynamics with a runtime scaling as $O(\min \{1/g \log(g N), \sqrt{N}\})$, achieving exponentially faster rates than previous results \cite{childs2016optimal}. In \cite{Kooper}, de Lacy also uses the Gross-Pitaevskii nonlinearity, performing an unstructured search on a complete graph in time  $\mathcal{O}(N/(g\sqrt{m(N-m)}))$, where $m$ is the number of marked terms. Others have applied the concept of nonlinear quantum dynamics to quantum walks \cite{Jason1,Jason2,Jason3,JasonExp}.

For the purposes of our algorithm, the key piece of literature is Childs and Young's work \cite{childs2016optimal}, which explicitly details the capability of nonlinear quantum computing to distinguish between two extremely finely separated single-qubit quantum states (on the order of $1/\exp(n))$). We expand upon their work and implement a new optimisation routine to provide an efficient quantum algorithm for calculating the maximum edge overlap.

\subsection{Groundwork for algorithm: postselection and candidate states}\label{sec:candid}

Before describing the nonlinear mechanism for distinguishing between two single-qubit states, let us describe how such a protocol is relevant to calculating the maximum edge overlap. Start with the map introduced in Section \ref{sec:quant}. From Eq.\cref{eqn:sec2}, proceed by performing the inverse of $G$ on the first register (as depicted in Fig. \ref{fig:hadamardtest}). The effect of this is
\begin{align*}
	\ket{0}^{\otimes \mathcal{O}(\log(n))}& \otimes \ket{1} \ket{2} \dots \ket{n} \otimes \frac{1}{\sqrt{n!}} \bigg(\sum_{\eo(\sigma_i) \leq E} \ket{\sigma_i}_{\mathrm{radix}}\ket{0} + \sum_{\eo(\sigma_i) > E} \ket{\sigma_i}_{\mathrm{radix}} \ket{1} \bigg) \\ \mapsto &\ket{0}^{\otimes \mathcal{O}(\log(n))} \otimes \ket{1} \ket{2} \dots \ket{n} \otimes \bigg[G^{\dagger} \bigg( \frac{1}{\sqrt{n!}} \sum_{\eo(\sigma_i) \leq E} \ket{\sigma_i}_{\mathrm{radix}}\bigg)\ket{0} + G^{\dagger}\bigg(\frac{1}{\sqrt{n!}} \sum_{\eo(\sigma_i) > E} \ket{\sigma_i}_{\mathrm{radix}} \bigg)\ket{1}\bigg]\\
     = &\ket{0}^{\otimes \mathcal{O}(\log(n))} \otimes \ket{1} \ket{2} \dots \ket{n} \otimes \bigg[ \bigg(\frac{n!-m}{n!} \ket{0}^{\otimes \Theta(n\log(n))} + \ket{\Psi'}\bigg)\ket{0} +\bigg(\frac{m}{n!} \ket{0}^{\otimes \Theta(n\log(n))} +\ket{\Psi''}\bigg)\ket{1} \bigg]\, ,
\end{align*}
where $m$ is the number of permutations having edge overlap greater than $E$. $\ket{\Psi'}$ and $\ket{\Psi''}$ are unknown, unnormalised first-register quantum states orthogonal to $\ket{0}^{\Theta(n\log(n))}$ satisfying $|\braket{\Psi'}|=\sqrt{\frac{n!-m}{n!}(1-\frac{n!-m}{n!})}$ and  $|\braket{\Psi''}|= \sqrt{\frac{m}{n!}(1-\frac{m}{n!})}$.

Following in the footsteps of \cite{abrams1998nonlinear}, we postselect the first register by the state $\ket{0}^{\otimes \Theta(n\log(n))}$ in order to induce a useful single-qubit quantum state in the ancilla. This postselection succeeds with probability at least 1/2 -- or, more precisely,
\[P(\text{postselection success}) = 1-\frac{2m}{n!}+2\bigg(\frac{m}{n!}\bigg)^2\, ,\]
and from this point onwards we will safely assume that the postselection always succeeds. The resulting single-qubit state after postselection by $\ket{0}^{\otimes \Theta(n\log(n))}$ is
\begin{equation}\label{eqn:houtput}
\frac{n!}{\sqrt{(n!)^2-2(n!)m+2m^2}}\bigg(\frac{n!-m}{n!}\ket{0} + \frac{m}{n!}\ket{1}\bigg)\, .
\end{equation}
The magnitude of the inner product of \cref{eqn:houtput} with $\ket{0}$ is
\begin{equation}
\label{eqn:overlap} \frac{n!-m}{\sqrt{(n!)^2-2(n!)m+2m^2}} = 1 - \frac{1}{2}\bigg(\frac{m}{n!}\bigg)^2+\mathcal{O}\bigg(\frac{m}{n!}\bigg)^3\, .
\end{equation}

One can draw two conclusions from Eq.\cref{eqn:overlap}: firstly, that for each value of $m \in [0,n!]$, the single-qubit output state in Fig. \ref{fig:hadamardtest} is different. Secondly, for $m=0$ the output state is simply $\ket{0}$, but as $m$ proceeds towards $n!$ the output moves strictly further away from $\ket{0}$ towards $\ket{1}$ along the real arc of the Bloch sphere (the semicircle $x^2+z^2=1$, $x\geq 0$). When $m=n!$, the output state is $\ket{1}$. We will refer to the set of all $n!+1$ different output states the set of \textit{candidate states}, and we will term the specific state for which $m=i$ as the $i$th candidate state. Because the amplitudes in the output state are not complex, we can also deduce that they all lie along the same arc on the Bloch sphere. %(specifically, the semicircle $x^2+z^2=1$ with $x>0$).

\begin{center}
\begin{figure}[H]
	\centering
    \includegraphics[scale=0.9]{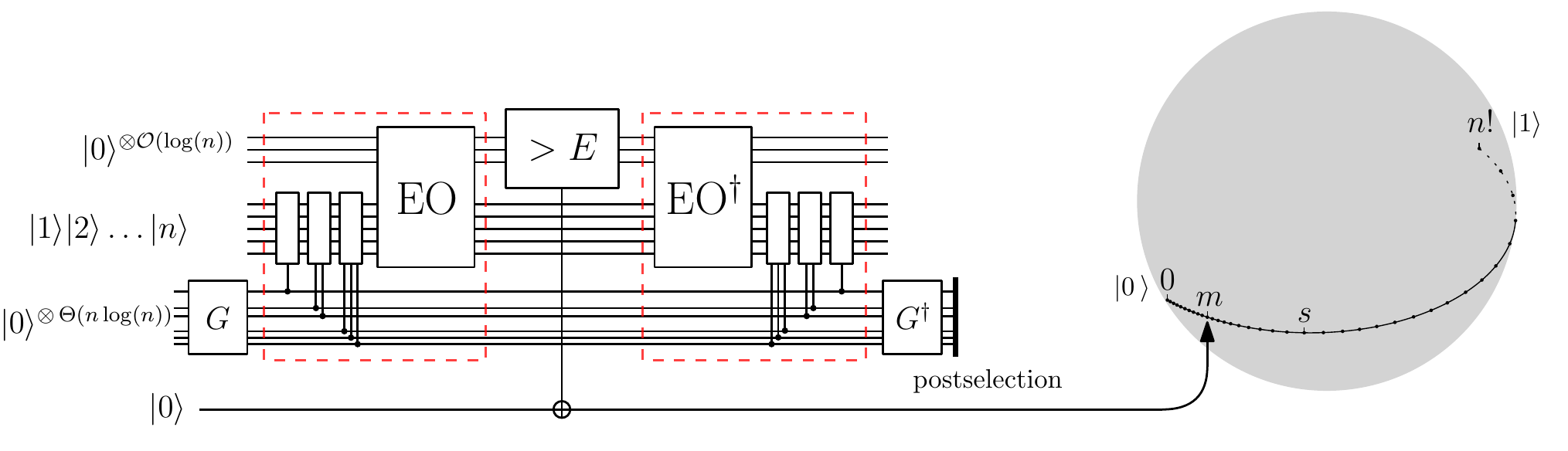}
    \caption{Our circuit produces one of $n!+1$ candidate states depending on $m$, the number of vertex permutations $\sigma$ with $\eo(\sigma) > E$. All candidate states lie along the same arc of the Bloch sphere. For all $s \in [0,n!]$, we know the exact coordinates of $s$th candidate state.}
    \label{fig:hadamardtest}
\end{figure}
\end{center}

The important task that we have achieved via this construction is encoding information about the edge overlap threshold $E$ into a single-qubit state. If $E$ is less than the maximum edge overlap, there will be a nonzero number of marked states (as seen in Fig. \ref{fig:threshold}) and thus the output of the circuit in Fig. \ref{fig:hadamardtest} will be one of the candidate states for which $m >0$.
If we could distinguish between the cases $m=0$ and $m>0$ efficiently, then we could construct an efficient algorithm not unlike Algorithm \ref{alg:durr} to find $\mathbf{MEO}(G_1,G_2)$. Measuring the output of the circuit in Fig. \ref{fig:hadamardtest} is only useful to distinguish between the cases $m=0$ and $m=n!$, and results in
\begin{equation}\label{eqn:distextrem}
	\begin{cases} \ket{0} & m = 0 \\ \ket{1} & m = n!\, . \end{cases}
\end{equation}

The remainder of this section details how the nonlinear quantum search algorithm described in \cite{childs2016optimal} can be used on the output of the circuit in Fig. \ref{fig:hadamardtest} to efficiently determine if $m=0$ or $m>0$. %Thus a protocol for single-qubit state discrimination forms a crucial part of our algorithm for calculating maximum edge overlap. 

\subsection{Nonlinear Quantum Search to Distinguish Candidate States}\label{sec:disting}

Nonlinear quantum computing makes use of systems whose time evolution can be approximated by the nonlinear Schr\"odinger equation,
\begin{align}\label{eqn:nlse}
	\ii  \partial_t | \varphi \rangle = (H + K ) | \varphi \rangle \, ,
\end{align}
where \(H\) is a typical time-dependent Hermitian operator as seen in linear quantum computing. The nonlinearity of the system is introduced by the operator $K$. Nonlinear quantum search was first proposed by Abrams and Lloyd \cite{abrams1998nonlinear}.
%By depicting a series of linear transformations and single-qubit nonlinear quantum transformations abiding by a nonlinear Schr{\"o}dinger equation, they provide a nonlinear search capable of solving NP-complete and \#P problems efficiently.
More recently, Childs and Young \cite{childs2016optimal} developed a continuous-time search based upon the same properties of nonlinear quantum mechanics, using Eq.\cref{eqn:nlse} as their specific nonlinear Sch{\"o}dinger equation. In both \cite{abrams1998nonlinear} and \cite{childs2016optimal}, it is duly noted that the marvel of nonlinear quantum mechanics is in its capacity for efficient state discrimination between two exponentially (and even factorially) close states in polynomial time, which was shown not to be attainable by linear quantum computing in Section \ref{sec:grov}.

Following in the footsteps of Childs and Young \cite{childs2016optimal}, we will consider a nonlinear evolution satisfying
\begin{align}\label{eqn:gpn}
	\langle x | K | \varphi \rangle = g | \langle x | \varphi \rangle |^2 \langle x | \varphi \rangle \, ,
\end{align}
where \(g \in \mathbb{R}\) is the nonlinearity strength. Eq.\cref{eqn:gpn}, with its cubic nonlinearity, is known specifically as the Gross-Pitaevskii nonlinearity. Eq.\cref{eqn:nlse} is a statistical model of the behaviour exhibited by systems such as Bose-Einstein condensates -- other nonlinear Schr\"odinger equations exist and are apt descriptions of different systems, such as Bose liquids \cite{meyer2014quantum}.
 
In \cite{childs2016optimal}, Childs and Young detail a nonlinear quantum algorithm which could be used to efficiently distinguish between what we referred to in Section \ref{sec:candid} as the $0$th and $1$st candidate states. They also mention that their protocol can be extended to efficiently distinguish between the $0$th candidate state and other given candidate state (say the $s$th candidate state). We make this mathematically explicit, and it is this process that underpins our main result: Algorithm \ref{alg:gcvnqs}. 

Algorithm \ref{alg:gcvnqs} relies on two procedures, which are fully detailed in Section \ref{sec:subrs}.
\begin{itemize}
	\item Procedure B, which produces a qubit in the $m$th candidate state (where $m$ is the number of permutations having edge overlap greater than the current threshold overlap $E$).
	\item Procedure A, which determines if $m=0$ or $m>0$. It makes use of Procedure B.
\end{itemize}

\subsection{Procedures for final algorithm}\label{sec:subrs}

\subsubsection{Procedure B: sending the $0$th and $s$th candidate states to opposite poles of the Bloch sphere}\label{sec:meas}

\textit{Foreword:}
This section discloses how to use the nonlinear evolution procedure from \cite{childs2016optimal} in order to send the $0$th and $s$th candidate states to the Bloch sphere poles $\ket{0}$ and $\ket{1}$, respectively. The reader might note that the linear quantum circuit in Fig. \ref{fig:hadamardtest}, which is based on discussion in \cite{abrams1998nonlinear}, deviates slightly from the one given in \cite{childs2016optimal}. The choice of circuit affects the distribution of candidate states on the Bloch sphere. The main difference that has informed our choice is:
\begin{itemize}
	\item In \cite{childs2016optimal}, candidate states are distributed along a small curve on the Bloch sphere. This curve is not an arc. The $0$th candidate state resides at $\ket{0}$ and the $(n!)$th lies at the end of this small curve. Childs and Young's algorithm can be used to distinguish between any two known candidate states.
    \item In this paper and in \cite{abrams1998nonlinear}, candidate states are distributed along the semicircle $x^2+z^2=1$, $x \geq 0$ on the Bloch sphere. The $0$th candidate state resides at $\ket{0}$ and the $(n!)$th lies at $\ket{1}$.
\end{itemize}
Although using the nonlinear quantum procedure from \cite{childs2016optimal}, we have adopted the preceding linear circuit from \cite{abrams1998nonlinear} due to the simplicity of the candidate state curve it bestows, which de-clutters much of the ensuing complexity analysis. At this point, then, it is prudent to note that this section is phrased as the goal of distinguishing between the $0$th and $s$th candidate states on the candidate state arc from \cite{abrams1998nonlinear}, although it follows the methodology of \cite{childs2016optimal}.

\bigskip

\begin{algorithm}[H]
\caption*{\textbf{Procedure B:} candidate state generation and orientation}\label{alg:srb}
\begin{algorithmic}[0]
\INPUT{$n, 1\leq s \leq n!, 0<m<s, G_1, G_2$}
	\State generate a qubit in the $m$th candidate state via the circuit in Section \ref{sec:candid}.
    \State orient the arc of candidate states on the Bloch sphere;
    \State run nonlinear evolution for time $T(s/n!)$
    \State orient the curve of post-evolution candidate states so the $0$th and $s$th are at $\ket{0}$ and $\ket{1}$, respectively.
\OUTPUT{a single-qubit state ready for measurement in Procedure B}
\end{algorithmic}
\end{algorithm}

\textit{Complexity summary:} As discussed in Section \ref{sec:candid}, generating a qubit in the $m$th candidate state takes $\mathcal{O}(n^2 \log(n))$ fundamental quantum gates. Both ``orient'' steps are simple rotations on the Bloch sphere, and so take a negligible amount of $\mathcal{O}(1)$ fundamental quantum gates. The nonlinear evolution time, $T(s/n!)$, is at most $\mathcal{O}(\frac{1}{g} \log(n!)) = \mathcal{O}(\frac{1}{g} n \log(n))$ which demonstrates the crucial feature of nonlinear quantum mechanics: it can distinguish between factorially-finely spaced states in polynomial time. Procedure B uses $\mathcal{O}(n^2 \log(n))$ fundamental quantum gates and nonlinear evolution for time $T(s/n!)$.

\bigskip

\textit{Procedure B analysis:} Let $S$ be the sub-arc of the original candidate state arc, stretching from the $0$th candidate state to the $s$th. Using linear quantum mechanics, we can orient the arc to the position shown in Fig. \ref{fig:coords} (this initial placement of the arc is justified soon). We will use nonlinear quantum mechanics to stretch the endpoints of this arc to opposite ends of the Bloch sphere -- and then an idealised measurement of the qubit will result in $\ket{0}$ if $m=0$ and $\ket{1}$ if $m =s$.

\begin{figure}[H]
	\centering
    \includegraphics[]{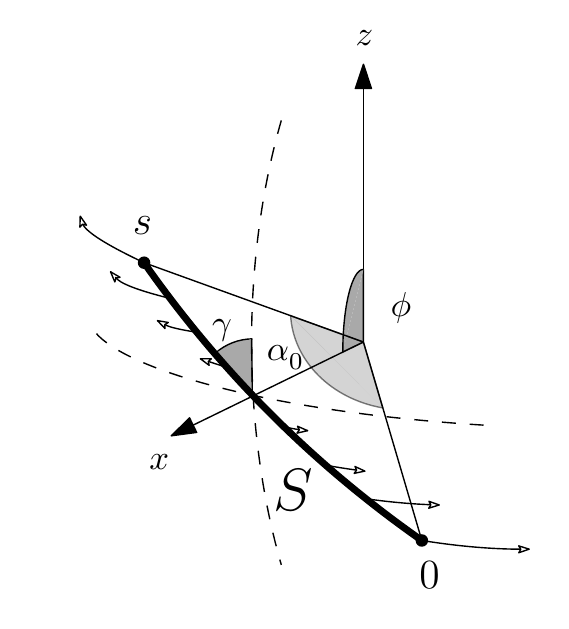}
    \caption{The position of $G$, the arc containing the candidate states on the surface of the Bloch sphere, and definition of angles $\phi$, $\alpha_0$ and $\gamma$. The endpoints separate the fastest under Gross-Pitaevskii evolution when $\phi=\pi/2$ and $\gamma = \pi/4$. The distorting effect of nonlinear time evolution is indicated by the white-tipped arrows. Note that the $0$th candidate state is not located at the pole $\ket{0}$ at this stage.}
    \label{fig:coords}
\end{figure}

Consider the nonlinear time evolution of a qubit initially lying on $S$ under the Gross-Pitaevskii equation, Eq.\cref{eqn:nlse}. The qubit's density matrix, $\rho$, can be specified by
\begin{align}
	\rho = \frac{1}{2} 
	\left( \begin{array}{cc}
	1+z & x- \ii  y \\ x+ \ii  y & 1-z
	\end{array} \right) \, ,
\end{align}
where \(x,y,z : [0,T] \to \mathbb{R}\) describe the qubit's position in 3D space on the Bloch sphere, and \(T \in \mathbb{R}\). To analyse the purely nonlinear evolution due to $K$ in Eq.\cref{eqn:nlse}, let $H =0$ for the time being. \cite{childs2016optimal} show that the qubit's evolution across the sphere simplifies to
\begin{align}\label{eqn:diffnlse}
	\text{d}_t (x,y,z) = gz (-y,x,0) \, .
\end{align}
The distorting effect of Eq.\cref{eqn:diffnlse} on $G$ is sketched in Fig. \ref{fig:coords}.

%Evolution of \(G\) via the linear Schr\"odinger equation preserves the inner product and performs a homeotopically trivial transformation on \(\mathrm{U(2)}\). Transformations abiding by linear quantum mechanics can therefore be used to orient \(G\) in the most practical way. We will find that the optimal orientation for the boundaries of \(G\) is similar to the orientation of the two states considered by Childs and Young \cite{childs2016optimal}. 

As in Fig. \ref{fig:coords}, let $\phi$ be the polar angle of the midpoint of $S$ and $\gamma$ be the anticlockwise rotation of $S$ from the line of latitude. Let $S(t)$ refer to the distorted curve after nonlinear evolution for time $t$. Finally, let the angle subtended by the endpoints of the curve at time $t \in [0,T]$ to be \(\alpha(t)\), and denote $\alpha (0)$ by $\alpha_0$. Note that, via Eq.\cref{eqn:houtput}, we have
\begin{equation}\label{eqn:a0}\cos\bigg(\frac{\alpha_0}{2}\bigg) = \frac{n!-s}{\sqrt{(n!)^2-2(n!)s+2s^2}} \, .\end{equation}
In Hilbert space (not the Bloch sphere), the inner product of the states at the endpoints of $S(t)$ is $\cos(\alpha(t)/2)$. Childs and Young show that this quantity decreases fastest when \(\phi=\pi/2\) and \(\gamma = 3\pi /4\), which is why this orientation appears in Fig. \ref{fig:coords}. As shown in \cite{childs2016optimal}, the inner product of the endpoints abides by
\begin{align} \label{eqn:angle}
	\cos \left(\frac{\alpha(t)}{2}\right) =  \frac{\cos\frac{\alpha_0}{2}  \cosh \frac{gt}{2}  -  \sinh\frac{gt}{2}}{\cosh \frac{gt}{2} - \cos \frac{\alpha_0}{2} \sinh\frac{gt}{2}}      \, .
\end{align}
The extremal states become orthogonal when Eq.\cref{eqn:angle} vanishes, i.e. we should set the evolution time to be
\begin{align}
T(s/n!) &= \frac{2}{g} \log \left( \cot \frac{\alpha_0}{4}  \right) \label{eqn:end_time} \\ & = \frac{2}{g} \bigg(\log\bigg(\frac{2n!}{s}\bigg) - \bigg(\frac{s}{n!}\bigg) +\mathcal{O}\bigg(\frac{s}{n!}\bigg)^2\bigg) \nonumber\\
	&\leq \mathcal{O}\bigg(\frac{1}{g} \log(n!) \bigg)\, . \nonumber
\end{align}
Nonlinear evolution elongates $S(t)$ laterally, thus immediately altering $\gamma$ from its ideal value of $\pi/4$. To counter this, we reintroduce the linear Hamiltonian term $H$ to Eqn.\cref{eqn:diffnlse}. Childs and Young \cite{childs2016optimal} show specifically that we can use the rotation
\begin{align}\label{eqn:lin_Hamiltonian}
	H(t) = \frac{g}{4} \cos (\frac{\alpha(t)}{2}) \sigma_x
\end{align}
to force the endpoints of \(S(t)\) to remain along the extension of the original arc $S$, and thus maintain $\gamma = \pi/4$. Here \(\sigma_x\) is the Pauli-\(x\) matrix. Note that the dilation of $S(t)$ is nonlinear, so non-endpoint candidate states will generally not lie on the extension of the original arc for any time $t>0$ (as seen in Fig. \ref{fig:lemma2}). After nonlinear evolution for time $T$, we can rotate the dilated curve $S(T(s/n!))$ via a simple unitary transformation so that the $0$th candidate state is at $\ket{0}$ and the $s$th is at $\ket{1}$, in much the same way we initially oriented $S$ in Fig. \ref{fig:coords}.

\bigskip

\subsubsection{Procedure: determining if $m= 0$ or $m > 0$}\label{sec:zoom}

\textit{Intuition:} Consider what would happen if we were to measure the qubit from Section \ref{sec:meas}. If $0<m<s$ then a single measurement will tell us nothing about the value $m$. A more meaningful measurement could be taken on an ensemble of $\omega$ of these qubits, resulting in the measurement outcome
\begin{equation}\label{eqn:cases2}\begin{cases} \ket{0}^{\otimes \omega} & \textrm{further investigation required} \\ \text{at least one } \ket{1} & m > 0\, .  \end{cases}\end{equation}
That is, measuring a single 1 implies $m>0$ (but the inverse proposition, regarding the measurement of all 0s, does not necessarily hold). In that event we should choose a smaller threshold energy, rebuild the circuit in Fig. \ref{fig:hadamardtest} accordingly, and repeat the process. \begin{comment}However, if we measured a string of 0s, we cannot assume that $m=0$ as it is still possible that $m>0$. Procedure A is a new protocol describing how to proceed if we measure all 0s: reduce the value of $s$ by at least a factor of 2 before generating a new ensemble and performing another measurement. By zooming in on exponentially diminishing sub-arcs of the candidate state arc in this way, we can determine if $m=0$ or $m>0$ in a time that is efficient in $n$.
\end{comment}

\bigskip

\textit{Procedure:} The zooming procedure consists formally of the below steps, which determine if $m=0$ or $m>0$ for any given threshold edge overlap $E$ with some probability of failure.

\begin{algorithm}[H]
\caption*{\textbf{Procedure A:} determine if \underline{$m=0$} or \underline{$m>0$}}\label{alg:sra}
\begin{algorithmic}[0]
\INPUT{$n, 1 \leq s \leq n!, 0 \leq m < s, G_1,G_2,\omega$} 
        		\While{$s \geq 1$}
              	\State  use Procedure B to generate $\omega$ qubits in the $m$th candidate state after nonlinear evolution.
                \Procedure{B: candidate state generation and orientation } {$n, s, G_1, G_2$} 
                	\State  generate a qubit in the $m$th candidate state via the circuit in Section \ref{sec:candid}.
                    \State orient the arc of candidate states on the Bloch sphere.
                    \State run nonlinear evolution for time $T(s/n!)$.
                    \State orient the curve of candidate states so the $0$th and $s$th are at $\ket{0}$ and $\ket{1}$, respectively.				
                    \EndProcedure
                \State measure each of the $\omega$ qubits and record the results.
                \If{any 1s are measured}
                	\State \Return \underline{$m>0$} 
                \ElsIf{only 0s are measured}
                	\State $s \rightarrow \lfloor s/2 \rfloor$ \Comment{\textit{Only point at which the algorithm can fail.}}
                \EndIf
            \EndWhile
            \State \Return \underline{$m=0$}. \Comment{\textit{Only happens if $s=0$.}}
\OUTPUT{either \underline{$m=0$} or \underline{$m>0$}.}
\end{algorithmic}
\end{algorithm}

\textit{Complexity summary:}  At the end of this section we show that Procedure A uses at most $ \omega \, (\mathcal{O}(n^2\log(n))(1+\log_2(s))$ fundamental quantum gates and nonlinear evolution for time $\mathcal{O}(\frac{1}{g} \log(2s)(\log_2(n!/\sqrt{s})+1))$. In Theorem \ref{thm:prob}, we justify that taking as small an ensemble size as $\omega = \mathcal{O}(\log \log(n))$ is sufficient for a working algorithm. Further discussion in Section \ref{sec:SandO} confirms that the cost of repeatedly using Procedure A, both in terms of fundamental quantum gates and nonlinear evolution time, will never exceed polynomial time in $n$ when finding $\mathbf{MEO}(G_1,G_2)$.

\bigskip

\textit{Procedure A analysis:} This analysis is a proof of concept of the scheme outlined above.

After applying Procedure B to all $\omega$ qubits, we measure each one. Choose some constant scaling factor $k$, with the requirement that $k \in [1/s,1/2]$ (at the end of this analysis we pick $k=1/2$ for reasons outlined in the proof of Theorem \ref{thm:zoomok}). If each measurement results in $\ket{0}$, we assume $\mathbf{m\leq \lfloor ks \rfloor}$, set $s = \lfloor ks \rfloor$ and repeat the procedure. Otherwise, if at least one measurement resulted in $\ket{1}$, we know with certainty that $\mathbf{m>0}$ and thus terminate the procedure.

At this point, there are two cases: $m \leq \lfloor ks \rfloor$ or $m > \lfloor ks \rfloor $. The measurement outcomes and appropriate assumptions to make are summarised in Measurements \ref{meas:meas1} and \ref{meas:meas2}, respectively.

\begin{meas}\label{meas:meas1} (Case $m \leq \lfloor  ks \rfloor$) Let $k \in [1/s,1/2] $ and define $\theta(T(s/n!))/2$ as in Fig. \ref{fig:overlap}. Measurement of the $\omega$ qubits results in
\begin{align}
	\begin{cases}
    	00 \dots 0, & \text{Continuation: report }\mathbf{m\leq \lfloor ks \rfloor }, \text{ probability at least } (\cos^2(\theta(T(s/n!))/2))^\omega\\
        \textrm{at least one } 1, & \text{Termination: report } \mathbf{m>0}, \text{ probability at most } 1- (\cos^2(\theta(T(s/n!))/2))^\omega\, .
    \end{cases}
\end{align}
\end{meas}

\begin{meas}\label{meas:meas2} (Case $m > \lfloor ks \rfloor$) Let $k \in [1/s,1/2] $ and define $\theta(T(s/n!))/2$ as in Fig. \ref{fig:overlap}. Measurement of the $\omega$ qubits results in
	\begin{align}
    \begin{cases}
    	00 \dots 0, & \text{Failure: report }\mathbf{m\leq \lfloor ks \rfloor}, \text{ probability at least } (\cos^2(\theta(T(s/n!))/2))^\omega\\
        \textrm{at least one } 1, & \text{Termination: report }\mathbf{m>0}, \text{ probability at most } 1- (\cos^2(\theta(T(s/n!))/2))^\omega\, .
    \end{cases}
    \end{align}
\end{meas}

The only time the procedure fails is if $m>\lfloor ks \rfloor$ and we record all zeroes, thus assuming $\mathbf{m \leq \lfloor ks \rfloor}$ via Measurement \ref{meas:meas2}.  For any value of $m>0$, a procedure that never fails will encounter Measurement \ref{meas:meas2} only once. In Theorem \ref{thm:prob}, we justify that simply by setting $\omega = \mathcal{O}(\log \log (n))$ we can make the probability of procedure failure arbitrarily low. Accordingly, we concern ourselves only with the legitimacy of the assumptions made in Measurement \ref{meas:meas1}.

\begin{figure}[H]
	\centering
    \includegraphics[width=0.7\linewidth]{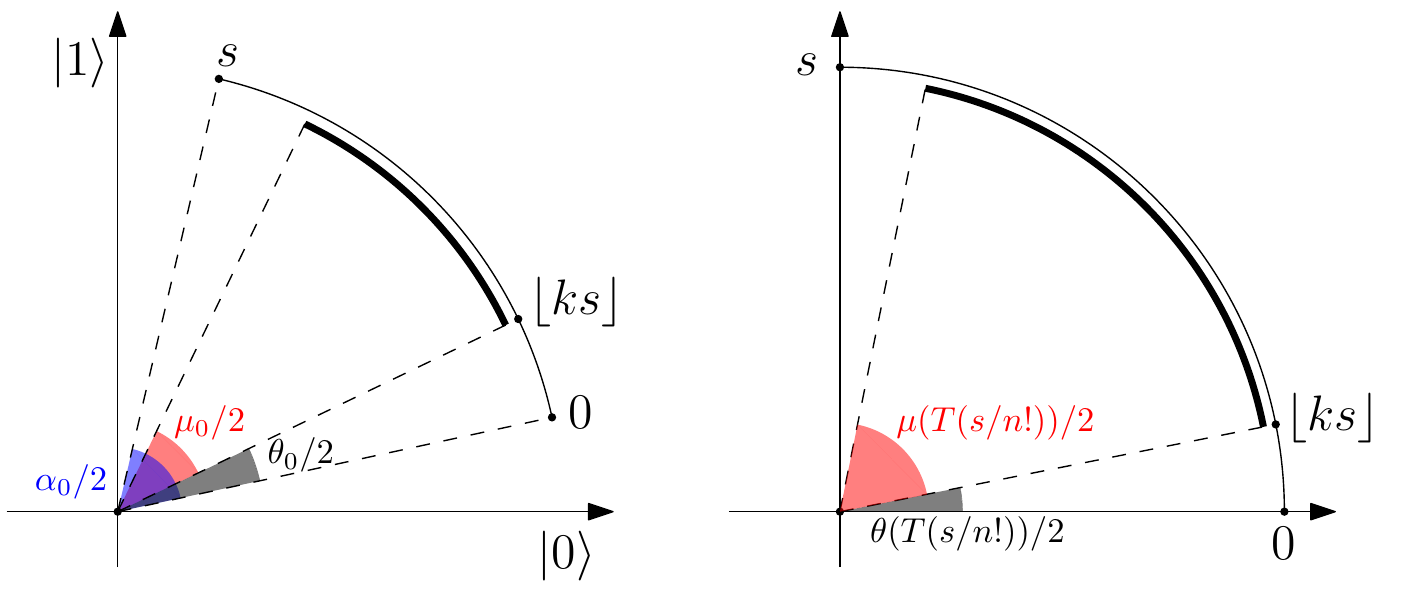}
    \caption{The validity of the zooming technique hinges on the correct choice of $k \in [1/s,1/2]$. This diagram describes the position of the $0$th, $\lfloor ks \rfloor$th and $s$th candidate states prior to nonlinear evolution at time $t=0$ (left) and afterwards, at time $T(s/n!)$ (right). Note: the $s$th and $0$th candidate states do not evolve directly to the states $\ket{1}$ and $\ket{0}$ as a result of nonlinear evolution, but are oriented onto these states afterwards by a simple rotation. For clarity, this diagram skips this step and simply displays $\ket{1}$ and $\ket{0}$ as the end result of nonlinear evolution.}
    \label{fig:overlap}
\end{figure}

If $m > \lfloor ks \rfloor$, we assume that Measurement \ref{meas:meas2} will result in a successful termination outcome and report $\mathbf{m>0}$; similarly, if $m \leq \lfloor ks \rfloor$ and any qubit is measured in the state $\ket{1}$, we will also terminate the procedure and report that $\mathbf{m>0}$. The technical part of this step is in justifying the continuation outcome for $m \leq \lfloor ks \rfloor$ in Measurement \ref{meas:meas1}. When we zoom by setting $s = \lfloor k s \rfloor$, recall that the candidate states are not evenly spaced. We allay any of the resulting concerns one might have about the zooming process via Theorem \ref{thm:zoomok}.

\begin{defn}\label{defn:suit} Let $n\geq 1$, $0\leq k \leq 1$ and $1 \leq s \leq n!$. The $\lfloor ks \rfloor$th candidate state is deemed ``suitable for zooming" if
\[\frac{1}{\sqrt{2}} \leq \cos\bigg(\frac{\theta(T(s/n!))}{2}\bigg) < 1 - \epsilon ,\]
where $\theta(T(s/n!))/2)$ is defined as in Figure \ref{fig:overlap} and $\epsilon = \mathcal{O}(1)$.
\end{defn}

\begin{thm} \label{thm:zoomok} For any $n$ and any $1 \leq s \leq n!$, we can efficiently find a value $k$ such that the $\lfloor ks \rfloor$the candidate state is suitable for zooming. Specifically,
\begin{enumerate}[label=(\alph*)]
\item For $1 \leq n \leq 5$, we can determine candidate states suitable for zooming on a case-by-case basis.
\item For $n \geq 5$ and $2 \leq s \leq n!$, the $\lfloor s/2 \rfloor $th candidate state is always suitable for zooming.
\end{enumerate}
\end{thm}

The proof of Theorem \ref{thm:zoomok} can be found in Appendix \ref{thm:zoomok}. As a consequence of this proof, we also have that $2\sqrt{2}/3$ is an upper bound for $\cos(\theta(T(s/n!))/2)$ for $n \geq 5$ and any $2 \leq s \leq n!$. As we have assumed that no incorrect measurements are made when Measurement \ref{meas:meas2} inevitably arises, we either successfully report that $\mathbf{m>0}$ after one of our many zooms, or let the algorithm run to completion indicating that $\mathbf{m=0}$. In the worst-case scenario, when $m$ is very small (i.e. $m$ is either $0$ or $1$), we will eventually set $s=1$. Our next step, distinguishing between the $0$th and $1$st candidate states, is what the algorithm in \cite{childs2016optimal} was built for, and we know that our measurement during this step will determine if $m=0$ or $m=1$ with certainty: the $\omega$ qubits in our measurement ensemble will all be in the state $\ket{1}$ if $m=1$, or $\ket{0}$ if $m=0$. This completes the analysis of the procedure.

Procedure B takes $\omega \, \mathcal{O}(n^2\log(n))$ fundamental quantum gates and exploits nonlinear evolution for time $\omega \, T(s/n!)$. The cost of Procedure A is then
\begin{align*}
	\omega \, \sum_{i=0}^{\lfloor \log_2s \rfloor} \bigg( \mathcal{O}(n^2 \log(n)) \quad &+\quad T(\lfloor s/2^i \rfloor)\bigg)\\
    \leq \quad \omega \, \bigg( \mathcal{O}(n^2\log(n))(1+\log_2(s)) \quad &+ \quad \sum_{i=0}^{\lfloor \log_2s \rfloor} T(\lfloor s/2^i \rfloor)\bigg)\, .
\end{align*}

\begin{lem}\label{lem:timecost}
For any $n\geq 2$ and $1\leq s \leq n!$, we have
	\[  \sum_{i=0}^{\lfloor \log_2s \rfloor} T(\lfloor s/2^i \rfloor) \leq \frac{2}{g} \log(2s) (\log_2(n!/\sqrt{s})+1) \, .\]
\end{lem}
The proof of Lemma \ref{lem:timecost} is in Appendix \ref{proof:timecost}. Therefore, this entire procedure takes time
\begin{align}\label{eqn:distinguish_time}
	\omega \, \bigg(\mathcal{O}(n^2\log(n))(1+\log_2(s)) +\frac{2}{g} \log(2s) (\log_2(n!/\sqrt{s})+1)\bigg)\, .
\end{align}

\section{Full algorithm for graph comparison} \label{sec:SandO}

At long last, we are armed with the ability to distinguish between the cases \(m=0\) and \(m > 0\). We can find $\mathbf{MEO}(G_1,G_2)$ via the following algorithm:

\begin{algorithm}[H]
\caption{Graph comparison via nonlinear quantum search}\label{alg:gcvnqs}
\begin{algorithmic}[0]
\INPUT{Graphs $G_1, G_2$ and $n\geq 5$. $\omega = \mathcal{O}(\log \log (n))$.}
\State $E\rightarrow E_{\max}$
\State $s\rightarrow n!$.
	\For{$i=1,2,\dots,\log_2{E_{\max}}$}
 		\State let $m$ be the number of permutations $\sigma \in S_n$ having edge overlap greater than $E$.
       \Procedure{ A: determine if \underline{$m=0$} or \underline{$m>0$} \,}  {$n,s, G_1,G_2,\omega$} \hspace{9.6cm} \makebox[0pt][r]{\tikzmark{start1}\phantom{\algorithmicprocedure}}
       		\While{$s \geq 1$}
              	\State  use Procedure B to generate $\omega$ qubits in the $m$th candidate state after nonlinear evolution.
                \Procedure{B: candidate state generation and orientation } {$n, s, G_1, G_2$} \hspace{2.8cm} \makebox[0pt][r]{\tikzmark{start2}\phantom{\algorithmicprocedure}}
                	\State  generate a qubit in the $m$th candidate state via the circuit in Section \ref{sec:candid}.
                    \State orient the arc of candidate states on the Bloch sphere.
                    \State run nonlinear evolution for time $T(s/n!)$.
                    \State orient the curve of candidate states so the $0$th and $s$th are at $\ket{0}$ and $\ket{1}$, respectively.				
                    \EndProcedure \makebox[0pt][r]{\tikzmark{end2}\phantom{\algorithmicend\ \algorithmicprocedure}}
                    \drawCodeBox{start2}{end2}
                \State measure each of the $\omega$ qubits and record the results.
                \If{any 1s are measured}
                	\State \Return \underline{$m>0$} 
                \ElsIf{only 0s are measured}
                	\State $s \rightarrow \lfloor s/2 \rfloor$ \Comment{\textit{Only point at which the algorithm can fail.}}
                \EndIf
            \EndWhile
            \State \Return \underline{$m=0$}. \Comment{\textit{Only happens if $s=0$.}}
            \EndProcedure \makebox[0pt][r]{\tikzmark{end1}\phantom{\algorithmicend\ \algorithmicprocedure}}
            \drawCodeBox{start1}{end1}
        \If {\underline{$m>0$}}
        	\State $E \rightarrow E + E_{\max}/2^i$
        \ElsIf {\underline{$m=0$}}
        	\State $E \rightarrow E - E_{\max}/2^i$
            \State $s \rightarrow n!$
        \EndIf
    \EndFor
\State \Return $E$
\OUTPUT{$E = \mathbf{MEO}(G_1,G_2)$ with probability greater than 1/2.}
\end{algorithmic}
\end{algorithm}

The following theorem proves that we can set $\omega = \mathcal{O}(\log \log(n))$ to ensure that Algorithm \ref{alg:gcvnqs} succeeds with constant probability.

\begin{thm}\label{thm:prob}
For $n \geq 5$, setting $\omega = \mathcal{O}(\log \log (n))$ is sufficient for Algorithm \ref{alg:gcvnqs} to succeed with probability greater than 1/2.
\end{thm}

\begin{proof}
Recall that the algorithm fails only if we assume $\mathbf{m \leq  \lfloor ks \rfloor}$ as a result of a Measurement \ref{meas:meas2}. Measurement \ref{meas:meas2} will only occur if $m>0$, and in this case successful application of Procedure A will only produce a instance of Measurement \ref{meas:meas2} once. Via Theorem \ref{thm:zoomok}, the probability of such a measurement resulting in termination and thus succeeding is at least $1-(2\sqrt{2}/3)^{2\omega}$. The entirety of Algorithm \ref{alg:gcvnqs} will execute Procedure A at most $\log_2(E_{\max})$ times, and so we have
\begin{align*}
    P(\textrm{Algorithm \ref{alg:gcvnqs} does not fail})	&\geq \bigg(1 - \left(2\sqrt{2}/3\right)^{2\omega} \bigg)^{\log_2(E_{\max})}\\
    	&= 1 - \log_2(E_{\max}) \left(2\sqrt{2}/3\right)^{2\omega} + \mathcal{O}\left(\left(2\sqrt{2}/3\right)^{4\omega}\right)\, .
\intertext{The condition requiring the leading term to be less than 1/2 is}\frac{1}{2}&>\log_2(E_{\max})\big(2\sqrt{2}/3)^{2\omega} \\
\omega &> \frac{1}{2} \frac{\log(2 \log_2(E_{\max}))}{\log\left(1/\left(2\sqrt{2}/3\right)\right)}\\ &= \mathcal{O}(\log \log(n))\, ,
\end{align*}
\begin{comment}
%OLD:
\begin{align*}
    P(\textrm{Algorithm \ref{alg:gcvnqs} does not fail})	&\geq \bigg(1 - \left(\sqrt{2/3}+\epsilon_n\right)^{2\omega} \bigg)^{\log_2(E_{\max})}\\
    													&= 1 - \log_2(E_{\max}) \big(\sqrt{2/3}+\epsilon_n\big)^{2\omega} + \mathcal{O}\left(\left(\sqrt{2/3}+\epsilon_n\right)^{4\omega}\right)\, .
\intertext{The condition requiring the leading term to be less than 1/2 is}\frac{1}{2}&>\log_2(E_{\max})\big(\sqrt{2/3}+\epsilon_n)^{2\omega} \\
\omega &> \frac{1}{2} \frac{\log(2 \log_2(E_{\max}))}{\log(1/(\sqrt{2/3}+\epsilon_n))}\\ &= \mathcal{O}(\log \log(n))\, ,
\end{align*}
\end{comment}
as $E_{\max} \leq n(n-1)/2$. In Appendix \ref{proof:ensemble}, we show that $\omega \approx 10 \log\log(n)$ is sufficient for $n \geq 5$.

\end{proof}

\underline{\textit{Algorithm \ref{alg:gcvnqs} complexity analysis:}}

Break the complexity up into two terms: the linear evolution time and the nonlinear evolution time contributed by the procedure from Section \ref{sec:meas}. In full, the highest cost that Procedure A can have is in zooming from $s=n!$ to $s=1$, which uses $\omega \, \mathcal{O}(n^2\log(n))(1+\log_2(n!))$ fundamental quantum gates and nonlinear evolution for time $\omega \, \frac{2}{g} \log(2n!) (\log_2(\sqrt{n!})+1)$.
As discussed in Theorem \ref{thm:prob}, Procedure A will occur at most $\log_2(E_{\max})$ times. Therefore an upper bound for the complexity of Algorithm \ref{alg:gcvnqs} is
\begin{align*}
    &\log_2(E_{\max}) \, \omega \, \bigg(\mathcal{O}(n^2\log(n))(1+\log_2(n!)) +\frac{2}{g} \log(2n!) (\log_2(\sqrt{n!})+1)\bigg)\\
    = \quad  & \mathcal{O}(n^3 \log^3(n) \log\log(n)) \quad + \quad \mathcal{O}\left(\frac{1}{g} n^2 \log^3(n) \log\log(n)\right)\, ,
\end{align*}
using $\omega = \mathcal{O}(\log\log(n))$ and $\log_2(E_{\max}) = \mathcal{O}(\log(n))$. The first term, $\mathcal{O}(n^3\log^3(n) \log\log(n))$, refers to the total number of fundamental quantum gates required to implement all of the conventional, linear circuits required for this algorithm. The second term, $\mathcal{O}(\frac{1}{g} n^2 \log^3(n))$, specifies the total duration of single-qubit evolution under the nonlinear Gross-Pitaevskii equation.

\textit{Note:} The worst-case scenario for Algorithm \ref{alg:gcvnqs} is if the maximum edge overlap is $\mathbf{MEO}(G_1,G_2)=E_{\max}/2$, where $E_{\max}=\min\{|G_1|,|G_2|\}$, and if there is only one permutation with this edge overlap, with all other permutations having edge overlap less than $E_{\max}/2$. This is the situation in which Algorithm \ref{alg:gcvnqs} will need to execute Procedure A $\log_2(E_{\max})$ times.

\section{Conclusion}

In this paper, we present an efficiently scaling quantum algorithm that finds the maximum edge overlap between a pair of general graphs, achieving exponential speedup over existing classical methods. The algorithm makes use of a two-part quantum dynamic process: in the first part we encode information crucial for the comparison of the two graphs in a single qubit. This information is hidden in a vanishingly small component of the system, having amplitude factorially small in the number of graph vertices. Because of this, even quantum algorithms such as Grover's search are not fast enough to distill this compenent efficiently. In order to extract the information we call upon techniques in nonlinear quantum computing to provide the required speed-up. All up, the linear quantum circuit requires $\mathcal{O}(n^3 \log^3 (n) \log \log (n))$ elementary quantum gates and the nonlinear evolution under the Gross-Pitaevskii equation has a time scaling of $\mathcal{O}(\frac{1}{g} n^2 \log^3 (n) \log \log (n))$, where $n$ is the number of vertices in each graph and $g$ is the strength of the Gross-Pitaveskii non-linearity.

The non-linear component of our algorithm is able to distinguish between no solutions ($m=0$) and at least one solution ($m>0$). Future work will involve extending this to efficiently find the number of solutions $m$ -- that is, solving problems in the $\#P$ complexity class. Our algorithm can also be readily adapted to efficiently determine if two graphs are isomorphic. 

Further to this, the nonlinear quantum component of our algorithm can in principle be applied to efficiently solve computational problems in the NP optimisation class. By formulating the problem in the Ising model \cite{Ising}, the maximum energy can be constructed to correspond to the optimal solution(s) so that our algorithm can determine this solution amongst exponentially many other potential solutions. Hence, with this algorithm we illustrate the theoretical power of non-linear quantum computation to solve problems which are considered infeasible for classical computers and (linear) quantum computers alike.

\section{Acknowledgements}

We thank Jason Twamley and Lyle Noakes for valuable comments and discussions.

 \begin{comment}
 \newpage
\subsection{Example}

\textbf{numerical example on small graphs to go here}
 
 \begin{figure}[H]
\hbox{\hspace{1.7em} 
\begin{overpic}[width=0.43\textwidth,unit=1mm]
{optim3.pdf}
\put(-5,25){\rotatebox{90}{\(H (\sigma_i)\)}}
\put(40,-2){\(i\)}
\put(76,23){\(\tau_3\)}
\put(76,38){\(\tau_4\)}
\put(76,48){\(\tau_5\)}
\end{overpic}}
%\vspace{2em}
\caption{Simulation of the quantum optimisation protocol on the list of bijection energies using the method of Section \ref{sec:SandO}.}
\label{fig:max}
\end{figure}
 
% \end{center}
% \end{widetext}
 
\end{comment}

\bibliographystyle{unsrt}
\bibliography{ref.bib}

\appendix

\section{Proofs}\label{app:proofzoomok}

\subsection{Proof of Theorem \ref{thm:zoomok}}

\begin{thm*}For any $n$ and any $1 \leq s \leq n!$, we can find a value $k$ such that the $\lfloor ks \rfloor$the candidate state is suitable for zooming. Specifically,
\begin{enumerate}[label=(\alph*)]
\item For $1 \leq n \leq 5$, we can determine candidate states suitable for zooming on a case-by-case basis.
\item For $n \geq 5$ and $2 \leq s \leq n!$, the $\lfloor s/2 \rfloor $th candidate state is always suitable for zooming.
\end{enumerate}
\end{thm*}

\begin{proof}
Consider the candidate state number as a continuous variable. That is, define the $s$th candidate state to be the point along the candidate state arc which has inner product with the $0$th candidate state equal to \[\frac{1-(s/n!)}{\sqrt{1-2(s/n!)+2(s/n!)^2}}\, ,\]
even if $s$ is not an integer. 

\textbf{Lower bound.}
Consider the candidate states immediately prior to nonlinear evolution, at time $t=0$ (refer to the left diagram in Fig. \ref{fig:overlap}). Immediately after production by the circuit in Fig. \ref{fig:hadamardtest}, the inner product of the $0$th and $s$th candidate states is 
	\begin{equation}\label{eqn:alpha0}
    	\cos\bigg(\frac{\alpha_0}{2}\bigg)= \frac{1-(s/n!)}{\sqrt{1-2(s/n!)+2(s/n!)^2}}\, .
    \end{equation}
Similarly, the inner product of the $0$th and $(ks)$th candidate states is
	\begin{equation}\label{eqn:theta0}
    \cos\bigg(\frac{\theta_0}{2}\bigg)=\frac{1-(ks/n!)}{\sqrt{1-2(ks/n!)+2(ks/n!)^2}}\, .
    \end{equation}
Define $\mu_0/2$ to be the angle subtended by the $ks$th candidate state and its projection when reflected about the midpoint of the arc between the $0$th and $s$th candidate states. Via simple trigonometry, we have
	\begin{align}\label{eqn:mu0}
    \cos\bigg(\frac{\mu_0}{2}\bigg)&=\cos(\alpha_0/2-\theta_0) = \cos\bigg(\frac{\alpha_0}{2}\bigg)\cos(\theta_0) + \sin\bigg(\frac{\alpha_0}{2}\bigg)\sin(\theta_0)\, .
    \end{align}

Due to the orientation of the candidate state arc imposed in Fig. \ref{fig:coords}, candidate states will move to opposite poles of the Bloch sphere depending on whether they are closer to the $0$th candidate state or the $s$th. All angles and figures defined in this section thus far have intuitively assumed that the $(ks)$th candidate state is closer to the $0$th candidate state than the $s$th. With the assistance of Lemma \ref{lem:halfway}, we have found and imposed the upper limit of $1/2$ on $k$ in order to ensure the consistency of all figures and angles in this section.

\begin{lem}\label{lem:halfway}
If $k \leq 1/2$, then the $(k s)$th candidate state is closer to the $0$th candidate state than the $s$th for all $1 \leq s \leq n!$.
\end{lem}

\begin{proof}
The $ (k s )$th candidate state is closer to the $0$th candidate state than the $s$th -- so it subtends an angle of at most $\alpha_0/4$ with the $0$th candidate state. The inner product with the $0$th candidate state of such a point is
\begin{align*}
\cos \bigg(\frac{\alpha_0}{4}\bigg) &= \sqrt{\frac{1}{2}\bigg(1+\frac{1-(s/n!)}{\sqrt{1-2(s/n)+2(s/n!)^2}}}\bigg)\, .
\end{align*}
The candidate state having this inner product, the $s'$th, is such that
\begin{align*}
			\frac{1-(s'/n!)}{\sqrt{1-2(s'/n!)+2(s'/n!)^2}} &= \sqrt{\frac{1}{2}\bigg(1+\frac{1-(s/n!)}{\sqrt{1-2(s/n)+2(s/n!)^2}}}\bigg)\, ,
\intertext{giving}
		\frac{s'}{n!} &= \frac{1}{1+\sqrt{1-2(s/n!)+2(s/n!)^2}} \frac{s}{n!} \\ &\geq \frac{s/2}{n!}\, .
\end{align*}
Therefore the $(ks)$th candidate state is guaranteed to be closer to the $0$th candidate state than the $s$th as long as $k \leq 1/2$.
\end{proof}
Lemma \ref{lem:halfway} also tells us that $\cos(\theta(T(s/n!))/2) \geq 1/\sqrt{2}$ as long as $k \leq 1/2$, because the minimum possible value of $\cos(\alpha_0/2)$ is $1/\sqrt{2}$ (and the angle cannot increase throughout the evolution). For this result to be meaningful, we must restore the status of $s$ as a discrete variable. In Measurements \ref{meas:meas1} and \ref{meas:meas2}, we take the $\lfloor ks \rfloor$th candidate state. In doing so, we pick either the $(ks)$th or $(ks-1)$th candidate state. In the latter case, we are effectively using $\tilde{k}=k-\frac{1}{s}$ as our $k$ value, which is less than $1/2$ and so our result still holds.

\textbf{Upper bound.} \textit{\underline{Aside:}} First, let us make it abundantly clear that the position of the ($ks$)th candidate \textit{after} nonlinear evolution for time $T(s/n!)$ discussed throughout this section (and in all figures referring to the position of the ($ks$)th candidate state after time $T(s/n!)$) is not the actual position of the ($ks$)th candidate state, but an approximation. The true position of the ($ks$)th candidate state after nonlinear evolution for time $T(s/n!)$ is some distance away from the arc connecting the $0$th and $s$th candidate states, as the evolution distorts the arc from its initial shape (see Fig. \ref{fig:lemma2}). This displacement need not cause concern, as we are in pursuit of an upper bound for $\cos(\theta(T(s/n!))/2)$, not a lower bound. This aside and Fig. \ref{fig:lemma2} justify this statement, and are the only parts of the proof that make reference to this true position of the ($ks$)th candidate state.

Consider the assortment of candidate states after nonlinear evolution for time $T(s/n!)$. Recall that the nonlinear evolution is corrected by a linear amount (see Eq.\cref{eqn:lin_Hamiltonian}) such that only the endpoints (and midpoint) of $S$ remain along the extension of the arc $S$. The orientation of $S$ was such that along this arc, points separate more quickly than if they were aligned along any other arc of the Bloch sphere (see Fig. \ref{fig:coords}). Therefore, any non-endpoint, non-midpoint candidate state will leave this arc and drag behind it as the $0$th and $s$th candidate states are separated. After nonlinear evolution for time $T(s/n!)$ as per Procedure B, this leads to the distortion as displayed in the dark curve in Fig. \ref{fig:lemma2}.

\begin{figure}[H]
	\centering
    \includegraphics[scale=0.8]{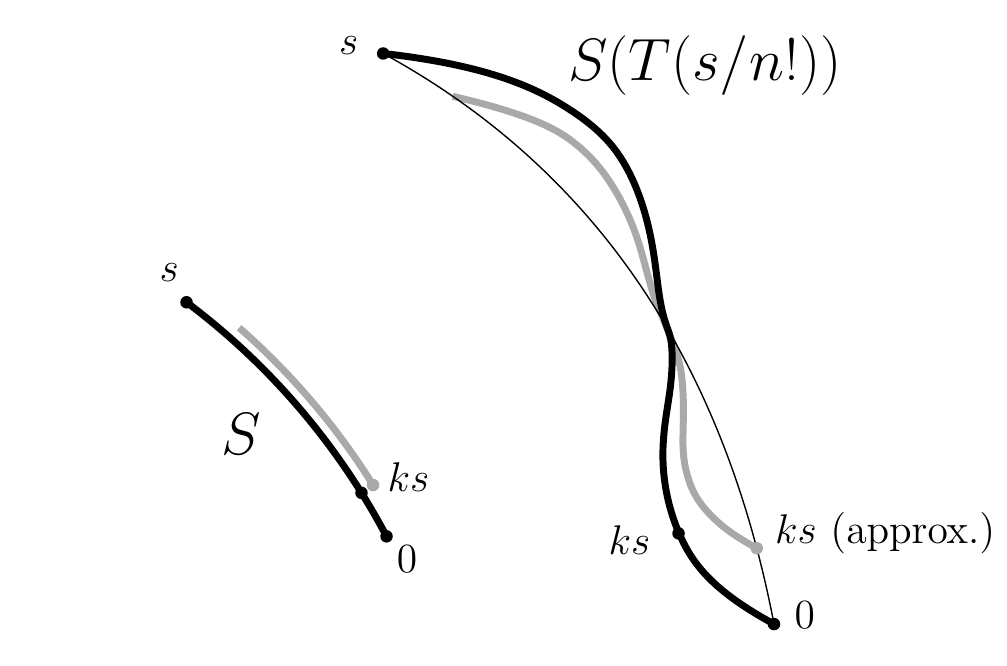}
    \caption{Left: prior to nonlinear evolution. Right: after nonlinear evolution for time $T(s/n!)$. The actual candidate states progress towards opposite poles more slowly than their approximate counterparts. This means that the $(ks)$th candidate state is further from the $0$th than the approximate $(ks)$th candidate state for all $k \leq 1/2$. We can determine the position of the approximate $(ks)$th candidate state by considering where the $(ks)$th candidate state would end up after nonlinear evolution for time $T(s/n!)$ if we were to parameterise the nonlinear evolution for an arc with endpoints initially subtending an angle of $\mu(t)$ (grey line), rather than $\alpha(t)$ (black line).}
    \label{fig:lemma2}
\end{figure}

There is a sub-arc of $S$ with the $(ks)$th candidate state as one endpoint and the other chosen symmetrically about the midpoint of $S$. Let the angle subtended by its endpoints in Hilbert space be $\mu(t)$, and let $\mu_0 =\mu(0)$. Consider the result if we performed the procedure of Section \ref{sec:meas} upon this sub-arc instead of $S$, but maintained the nonlinear evolution time of $T(s/n!)$. In this instance, Eq.\cref{eqn:angle} would become
\begin{align} \label{eqn:mut}
	\cos \left(\frac{\mu(t)}{2}\right) =  \frac{\cos\frac{\mu_0}{2}  \cosh \frac{gt}{2}  -  \sinh\frac{gt}{2}}{\cosh \frac{gt}{2} - \cos \frac{\mu_0}{2} \sinh\frac{gt}{2}}      \, .
\end{align}
and the nonlinear evolution, previously Eq.\cref{eqn:lin_Hamiltonian}, would instead be
\begin{align}\label{eqn:muham}
	H(t) = \frac{g}{4} \cos (\frac{\mu(t)}{2}) \sigma_x\, .
\end{align}
The nonlinear evolution time of $T(s/n!)$, however, is not long enough to force the endpoints of this smaller arc all the way to opposite poles of the Bloch sphere. Recalling that $T(s/n!) = 2\ln(\cot(\alpha_0/4))/g$, we have
\begin{align}\label{eqn:muts}
	\cos \bigg( \frac{\mu(T(s/n!))}{2} \bigg) &= \frac{\cos(\mu_0/2) - \cos(\alpha_0/2)}{1-\cos(\mu_0/2) \cos(\alpha_0/2)}\, .
\end{align}
In this hypothetical scenario, at the end of nonlinear evolution, this ``approximate" $(ks)$th candidate state would necessarily lie closer to $\ket{0}$ than the actual $(ks)$th candidate state does -- so the angle $\mu(T(s/n!))/2$ is greater than the angle subtended by the true $(ks)$th state and its symmetric point after nonlinear evolution for time $T(s/n!)$. Again, this works in our favour for proving Theorem \ref{thm:zoomok}. \textit{\underline{End aside.}}

    \begin{figure}[H]
        \begin{subfigure}[b]{0.475\textwidth}
            \centering
            \includegraphics[width=\textwidth]{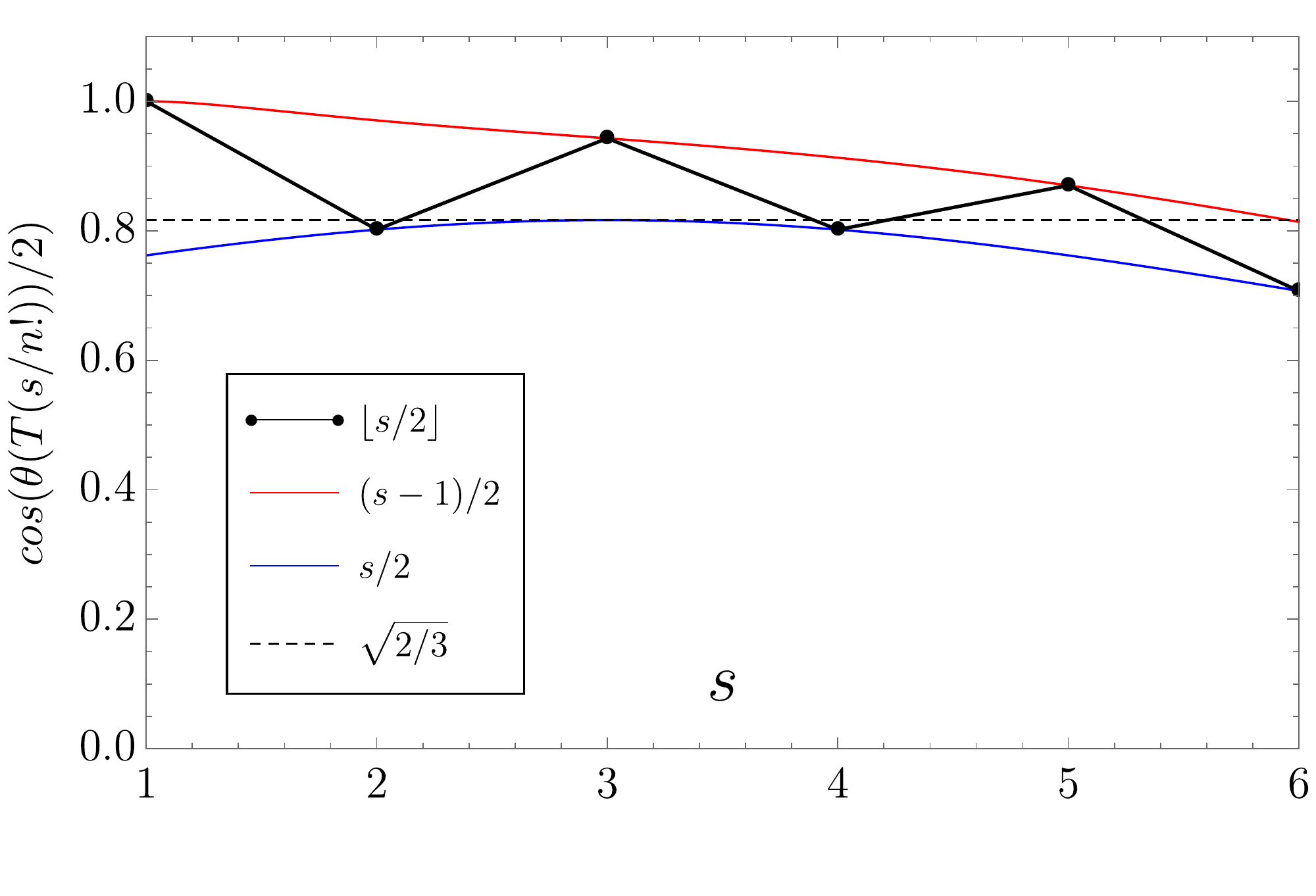}
            \caption*{$n=3$}%
            %{{\small Network 1}}    
            \label{fig:n_3}
        \end{subfigure}
        \quad
        \begin{subfigure}[b]{0.475\textwidth}  
            \centering 
            \includegraphics[width=\textwidth]{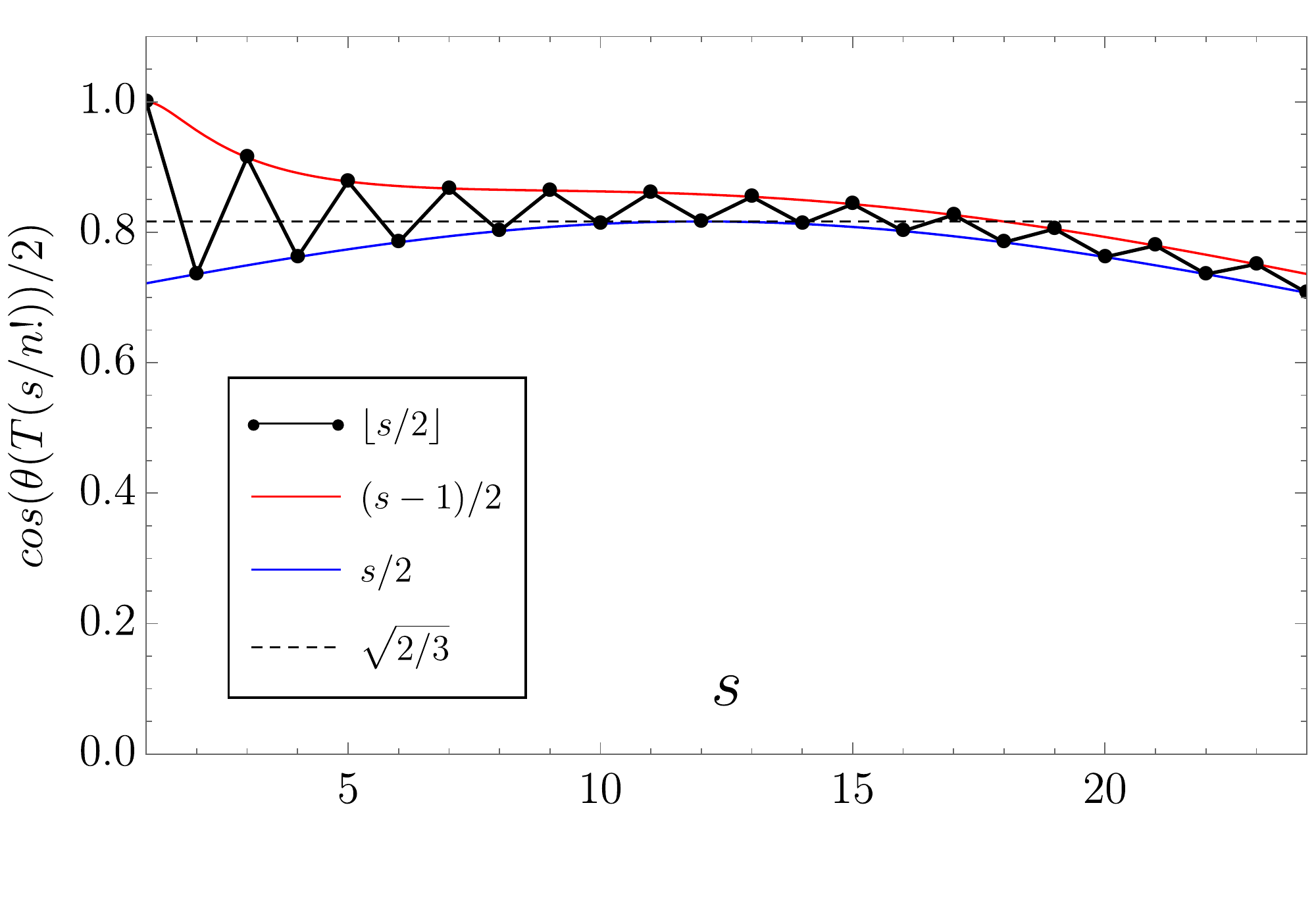}
            \caption*{$n=4$}
            %{{\small Network 2}}    
            \label{fig:n_4}
        \end{subfigure}
        \vskip\baselineskip
        \begin{subfigure}[b]{0.475\textwidth}   
            \centering 
            \includegraphics[width=\textwidth]{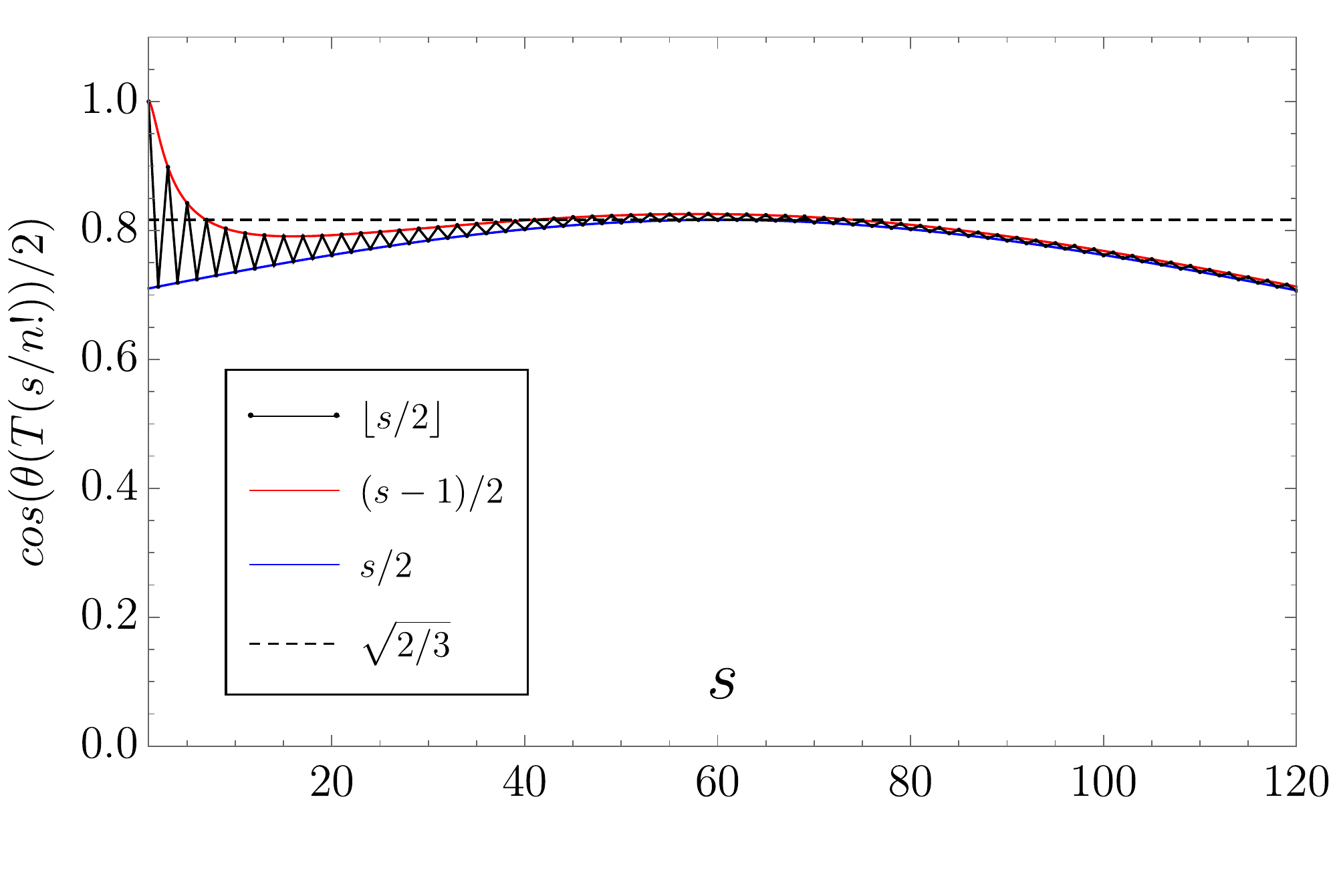}
            \caption*{$n=5$}
            %{{\small Network 3}}    
            \label{fig:n_5}
        \end{subfigure}
        \quad
        \begin{subfigure}[b]{0.475\textwidth}   
            \centering 
            \includegraphics[width=\textwidth]{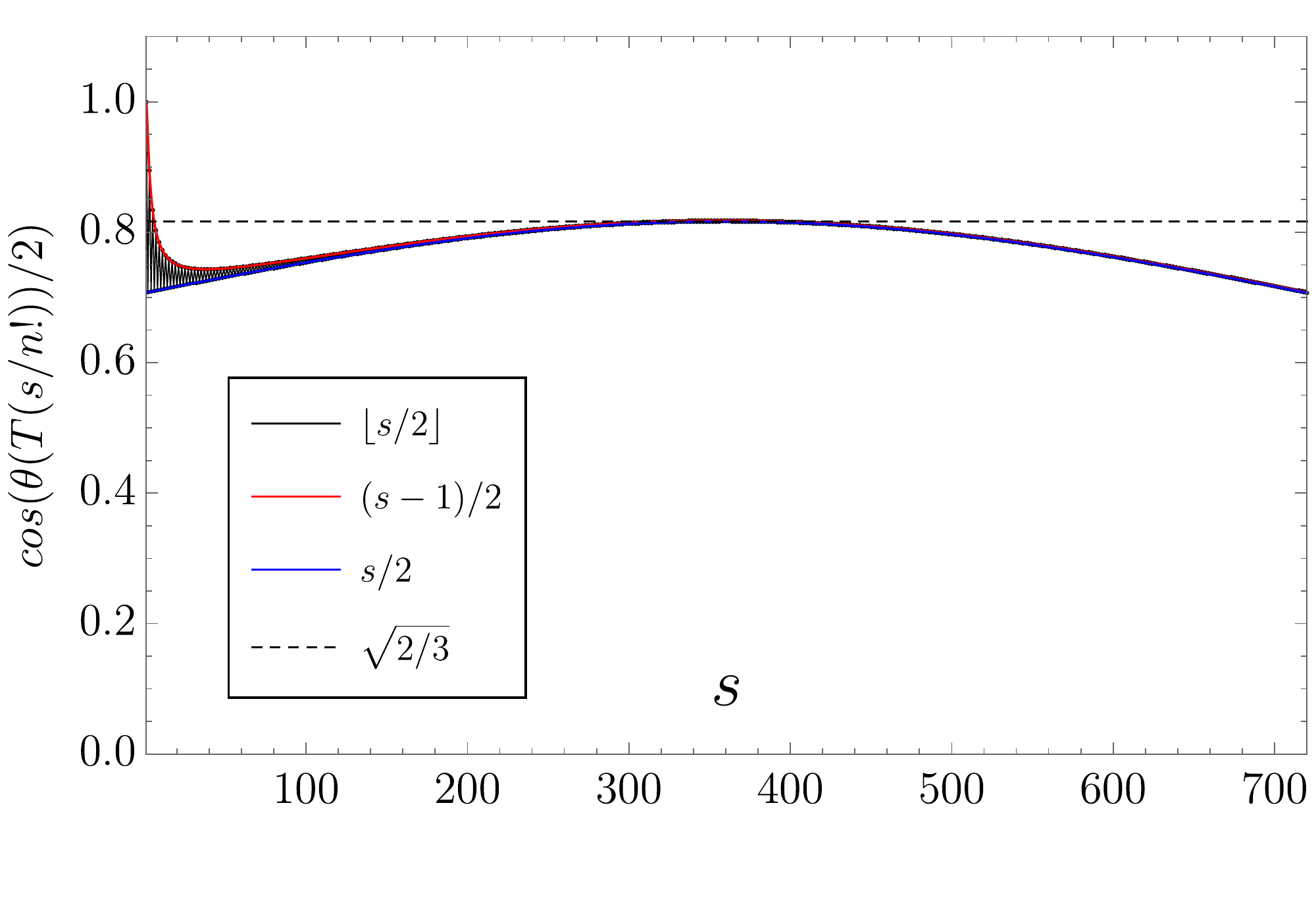}
            \caption*{$n=6$}%
            %{{\small Network 4}}    
            \label{fig:n_6}
        \end{subfigure}
        \caption{The data points are values of $\cos(\theta(T(s/n!))/2)$ for $1\leq s \leq n!$. The red curve is the upper bound, obtained when we replace $\lfloor s/2 \rfloor$ with $(s-1)/2$ in $\cos(\theta(T(s/n!))/2)$. Likewise, the blue curve is the lower bound we obtain by replacing $\lfloor s/2 \rfloor$ with $s/2$. For $n \geq 5$, we can safely assume that $\cos(\theta(T(s/n!))/2) \leq \sqrt{2/3} + \epsilon_n$ in the region $s>7$, where $\epsilon_n$ is a vanishingly small quantity corresponding to the maximum of the blue curve.}
        \label{fig:alldat}
    \end{figure}

Consider the assortment of candidate states after nonlinear evolution for time $T(s/n!)$. We orient this curve with the measurement basis $\{\ket{0},\ket{1}\}$ by aligning the $0$th and $s$th candidate states with $\ket{0}$ and $\ket{1}$, respectively. The curve is no longer an arc after nonlinear evolution, but for our purposes -- namely, the impending measurement of the qubit -- we need only consider the components of the candidate states in the $\ket{0}$ and $\ket{1}$ directions. The candidate state curve is then a quarter-circle when projected onto the measurement basis, as visualised in the right diagram in Fig. \ref{fig:overlap}. Via simple trigonometry, we have
\begin{align}
	\cos\bigg( \frac{\theta(T(s/n!))}{2}\bigg) &= \cos\bigg(\frac{\pi}{4} - \frac{\mu(T(s/n!))}{4}\bigg)\\
    				\label{eqn:thet2}		&= \frac{1}{2}\bigg(\sqrt{1+\cos(\mu(T(s/n!))/2)} + \sqrt{1- \cos(\mu(T(s/n!))/2)}\bigg)\, ,
\end{align}
which is a decreasing function of $\cos(\mu(T(s/n!))/2)$. Recall Eq.\cref{eqn:muts}:
\begin{align}
	\cos\left(\frac{\mu(T(s/n!))}{2}\right) &= \frac{\cos(\mu_0/2) - \cos(\alpha_0/2)}{1-\cos(\mu_0/2) \cos(\alpha_0/2)}\\
    										&= \frac{(\cos(\alpha_0/2)\cos(\theta_0)+\sin(\alpha_0/2)\sin(\theta_0))-\cos(\alpha_0/2)}{1-(\cos(\alpha_0/2)\cos(\theta_0)+\sin(\alpha_0/2)\sin(\alpha_0/2))\cos(\alpha_0/2)}\, ,
\end{align}
which, using Eqs.\cref{eqn:alpha0} and \cref{eqn:mu0} and basic trigonometry, 
\begin{equation}\label{eqn:muT}= \frac{\left(\frac{1-s/n!}{\sqrt{1-2s/n!+2(s/n!)^2}} \left(2\frac{(1-ks/n!)^2}{1-2ks/n!+2(ks/n!)^2}-1\right)+\frac{s/n!}{\sqrt{1-2s/n!+2(s/n!)}}\left(\frac{2ks/n!(1-ks/n!)}{1-2ks/n!+2(ks/n!)^2}\right)\right)-\frac{1-s/n!}{\sqrt{1-2s/n!+2(s/n!)^2}}}{1-\left(\frac{1-s/n!}{\sqrt{1-2s/n!+2(s/n!)^2}} \left(2\frac{(1-ks/n!)^2}{1-2ks/n!+2(ks/n!)^2}-1\right)+\frac{s/n!}{\sqrt{1-2s/n!+2(s/n!)}}\left(\frac{2ks/n!(1-ks/n!)}{1-2ks/n!+2(ks/n!)^2}\right)\right)\frac{1-s/n!}{\sqrt{1-2s/n!+2(s/n!)^2}}}\, .
\end{equation}
Note that, for $k=1/2$, this has a minimum of $2\sqrt{2}/3$ at $s=n!/2$. We substitute Eq.\cref{eqn:muT} into Eq.\cref{eqn:thet2} to obtain a large expression for $\cos(\theta(T(s/n!))/2)$. At this stage, it is convenient to restore the status of $s$ as a discrete variable by replacing $ks$ with $\lfloor s/2\rfloor$. Fig. \ref{fig:alldat} shows the values of $\cos(\theta(T(s/n!))/2)$ with $1 \leq s \leq n!$ for graphs of size $n=3,4,5$ and $6$, alternating between smooth functions which serve as upper and lower bounds:
\begin{align*}
	\cos\left(\frac{\theta(T(s/n!))}{2}\right)\bigg\vert_{\lfloor s/2\rfloor \longleftrightarrow s/2} \leq \cos\left(\frac{\theta(T(s/n!))}{2}\right) \leq \cos\left(\frac{\theta(T(s/n!))}{2}\right)\bigg\vert_{\lfloor s/2\rfloor \longleftrightarrow (s-1)/2}\, .
\end{align*}

Inspired by Fig. \ref{fig:alldat}, we consider only graphs with five or more vertices and show that $\cos(\theta(T(s/n!))/2)$ is always less than $\sqrt{2/3}+\epsilon$ (for a known quantity $\epsilon$) when $7 \leq s \leq n! $. When $s$ is even, $\cos(\theta(T(s/n!))/2)$ appears on the lower bound function (blue curve) in Fig. \ref{fig:alldat}, which has a maximum of $\sqrt{2/3}$ (using the fact that $\cos(\mu(T(s/n!))/2)$ from Eq.\cref{eqn:muT} has a minimum of $2\sqrt{2}/3$ at $s=n!/2$ when $k=1/2$). Therefore, for all even $s$,
\[\cos\left(\frac{\theta(T(s/n!))}{2}\right) \leq \sqrt{\frac{2}{3}}\, ,\]
and so we have shown that the $\lfloor s/2 \rfloor$th candidate state is suitable for zooming (via Definition \ref{defn:suit}) when $s$ is even.

When $s$ is odd, take the expression for $\cos(\theta(T(s/n!))/2)$ for fixed values of $s$ and varying values of $n$ to observe that
\begin{align}
	\cos\left(\frac{\theta(T(1/n!))}{2}\right) &= 1\, , \nonumber \\
    \frac{2}{\sqrt{5}} < \cos\left(\frac{\theta(T(3/n!))}{2}\right)& \leq \frac{2 \sqrt{2}}{3}\, , \label{eqn:s3} \\
    \frac{3}{\sqrt{13}} < \cos\left(\frac{\theta(T(5/n!))}{2}\right)& \leq \frac{18 \sqrt{2}}{29}\, , \label{eqn:s5}
    \intertext{and that, for $n \geq 5$,}
    \frac{4}{5} < \cos\left(\frac{\theta(T(7/n!))}{2}\right)& %\leq 80 \sqrt{\frac{2}{19209}}
    < \sqrt{\frac{2}{3}}\, . \label{eqn:slow}
\end{align}
As $s$ is odd, $\cos(\theta(T(s/n!))/2)$ lands on the upper bound function (red curve in Fig. \ref{fig:alldat}). For $n \geq 5$, the upper bound function has a maximum value, slightly above $\sqrt{2/3}$, occurring somewhere between $s = n!/2-1$ and $s=n!/2-2$. Eq.\cref{eqn:slow} assures us that the value of $\cos(\theta(T(s/n!))/2)$ is less than whatever this maximum is, as long as $s \geq 7$. We evaluate $\cos(\theta(T(s/n!))/2)$ at $s=n!/2-2$ to approximate the maximum as

\begin{equation}\begin{aligned}\frac{1}{2} &\left(\sqrt{\frac{\left(2 \sqrt{\frac{32}{(n!)^2}+2}+3\right) (n!)^4-4 \left(4 \sqrt{\frac{32}{(n!)^2}+2}+5\right) (n!)^3+12 \left(2 \sqrt{\frac{32}{(n!)^2}+2}+5\right) (n!)^2-192 n!+576}{n! (n! (n! (3 n!-20)+60)-192)+576}}\right.\\ & \quad +\left.\sqrt{\frac{\left(3-2 \sqrt{\frac{32}{(n!)^2}+2}\right) (n!)^4+4 \left(4 \sqrt{\frac{32}{(n!)^2}+2}-5\right) (n!)^3+12 \left(5-2 \sqrt{\frac{32}{(n!)^2}+2}\right) (n!)^2-192 n!+576}{n! (n! (n! (3 n!-20)+60)-192)+576}}\right)\, .\label{eqn:en}\end{aligned}\end{equation}
For the first two values of $n$ for which this maximum exists, we obtain
\begin{alignat*}{2}
	&n=5: \quad && \cos\left(\frac{\theta(T(s/n!))}{2}\right) \leq \sqrt{\frac{2}{3}} + \underbrace{0.00907762\dots}_{\epsilon_5}\\
    &n=6: \quad && \cos\left(\frac{\theta(T(s/n!))}{2}\right) \leq \sqrt{\frac{2}{3}} + \underbrace{0.00151206\dots}_{\epsilon_6}.
\end{alignat*}
We can calculate $\epsilon_n$ for any graph size $n$ by calculating the difference between \cref{eqn:en} and $\sqrt{2/3}$. However, as seen in Fig. \ref{fig:en}, $\epsilon_n$ is a decreasing function of $n$ and, for any $n \geq 5$, so it suffices to use $\sqrt{2/3}+\epsilon_5$ as an upper bound for $\cos(\theta(T(s/n!))2)$.

\begin{figure}[H]
	\centering
    \includegraphics[width=0.475\textwidth]{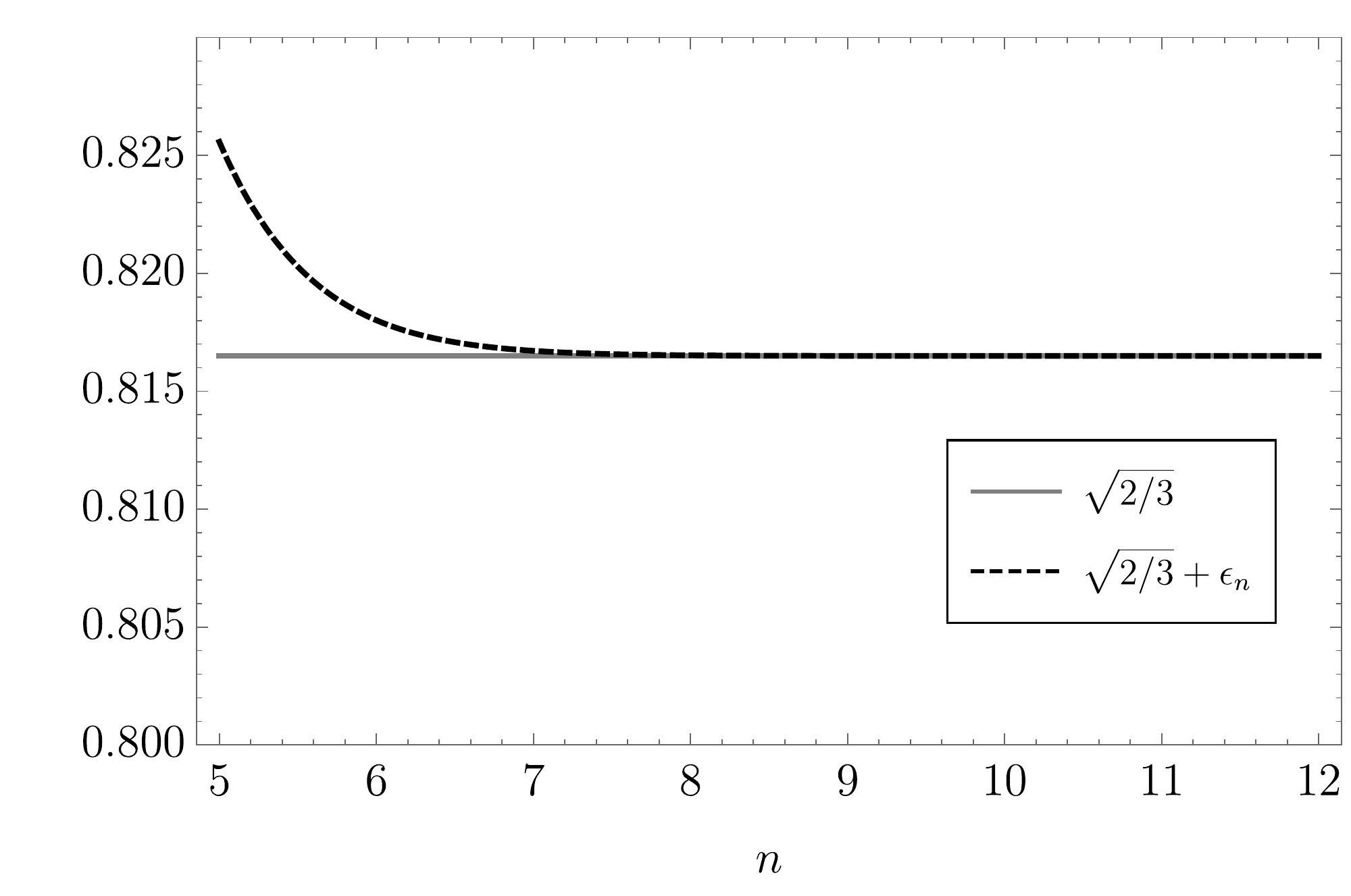}
    \caption{When $s$ is odd, the maximum value of $\cos(\theta(T(s/n!))/2)$ approaches $\sqrt{2/3}$ as $n$ increases.}
    \label{fig:en}
\end{figure}

Therefore we have shown that for $n \geq 5$ and odd $s \geq 7$ that the $\lfloor s/2 \rfloor$th candidate state is suitable for zooming (via Definition \ref{defn:suit}). The results in Eq.\cref{eqn:s3} and Eq.\cref{eqn:s5} also show that for $s=3$ and $s=5$, the $\lfloor s/2 \rfloor$th candidate state is also suitable for zooming. This completes the proof of part (b) of Theorem \ref{thm:zoomok}. Part (a) of Theorem \ref{thm:zoomok} is trivial: as there are only 153 pairs $(s,n)$ for $n \leq 4$, all different values of $\cos(\theta(T(s/n!))/2)$ will not draw arbitrarily close to 1. We can draw from a list of these 153 different values to determine which candidate states are suitable for zooming on a case-by-case basis. 

To summarise, for $ n \geq 5$, we have
\begin{alignat*}{3}
	&s = 1, && \quad \cos(\theta(T(s/n!))/2) = 1 &&\qquad\\
    &s = 2, && \quad 1/\sqrt{2} < \cos(\theta(T(s/n!))/2) < \sqrt{2/3} && \approx 0.82 \qquad \\
    &s = 3, &&\quad 1/\sqrt{2} < \cos(\theta(T(s/n!))/2) \leq 2\sqrt{2}/2 && \approx 0.94 \qquad\\
    &s = 4, &&\quad 1/\sqrt{2} < \cos(\theta(T(s/n!))/2) \leq \sqrt{2/3} \qquad \\
    &s = 5, &&\quad 1/\sqrt{2} < \cos(\theta(T(s/n!))/2) \leq 18\sqrt{2}/29 && \approx 0.88\qquad \\
    &s = 6, &&\quad 1/\sqrt{2} < \cos(\theta(T(s/n!))/2) \leq \sqrt{2/3}\\
    &s \geq 7, &&\quad 1/\sqrt{2} < \cos(\theta(T(s/n!))/2) \leq \sqrt{2/3} + \epsilon_5 && \approx 0.83\, .\\
\end{alignat*}

\end{proof}

\subsection{Proof of Lemma \ref{lem:timecost}}\label{proof:timecost}
\begin{proof}Using Eq.\cref{eqn:a0} and Eq.\cref{eqn:end_time}, we have
\begin{align*}
	T(s/n!) &=\frac{2}{g} \log(\cot(\frac{1}{2}\arccos(\frac{1-s/n!}{1-2s/n!+2(s/n!)^2})))\, .
\end{align*}
Treating $s/n!$ as a continuous  variable, we can take the Taylor series of $T(s/n!)$ about $0$ to give us an expression which is strictly less than or equal to $T(s/n!)$:
\begin{align*}
	T(s/n!) \leq \frac{2}{g} \left( \log \frac{2}{s/n!} - s/n! - \frac{(s/n!)^2}{4} + \frac{(s/n!)^3}{6}+\frac{13(s/n!)^4}{32} + \frac{17(s/n!)^5}{40} + \mathcal{O}(s/n!)^6 \right)
\end{align*}
By taking the series to the fifth order in $s/n!$, the above inequality holds. Then,
\begin{align*}
	\sum_{i=0}^{\lfloor \log_2 s\rfloor} T(\lfloor s/2^i\rfloor) &\leq \frac{2}{g}\sum_{i=0}^{\log_2 s} \left( \log \frac{2}{s/n!} - s/n! - \frac{(s/n!)^2}{4} + \frac{(s/n!)^3}{6}+\frac{13(s/n!)^4}{32} + \frac{17(s/n!)^5}{40} + \mathcal{O}(s/n!)^6 \right)\\
    &< \frac{2}{g}\left(\underbrace{\log\left(\frac{2s}{(s/(n!)^{\log 2})^{\log_2 2s}}\right)}_{O(\mathrm{poly}(n))} + \underbrace{\frac{68(s/n!)^5}{155}+\frac{13(s/n!)^4}{30}+\frac{4(s/n!)^3}{21}}_{\ll 1}\right)\\
    &\approx \frac{2}{g} \log(2s)\left(\log_2\frac{n!}{\sqrt{s}}+1\right)\, ,
\end{align*}
discarding the infinitesimal polynomial terms.\end{proof}

\subsection{Comment on Theorem \ref{thm:prob}: suitable choices for the ensemble size $\omega$}\label{proof:ensemble}
Recall the proof of Theorem \ref{thm:prob}. Procedure A has the highest chance of failure when $s=3$, as shown in Appendix \ref{proof:timecost}. During the course of an entire run-through of Algorithm \ref{alg:gcvnqs}, the iterations of Procedure B when $s = 3$ are infrequent enough that we will neglect their slight lowering of the success probability of Algorithm \ref{alg:gcvnqs}, and instead estimate the overall probability of success as
\[P(\textrm{Algorithm \ref{alg:gcvnqs} does not fail})	\approx \left(1 - \left(\sqrt{2/3}+\epsilon_n\right)^{2\omega} \right)^{\log_2(E_{\max})}\, . \]
Using this definition, we observe the pattern shown in Figure \ref{fig:psucc}. Most notably, the probability of Algorithm \ref{alg:gcvnqs} performing successfully is greater than $1/2$ for $\omega \geq 10 \log(\log(n))$.

\begin{figure}[H]
	\centering
    \includegraphics[width=0.475\textwidth]{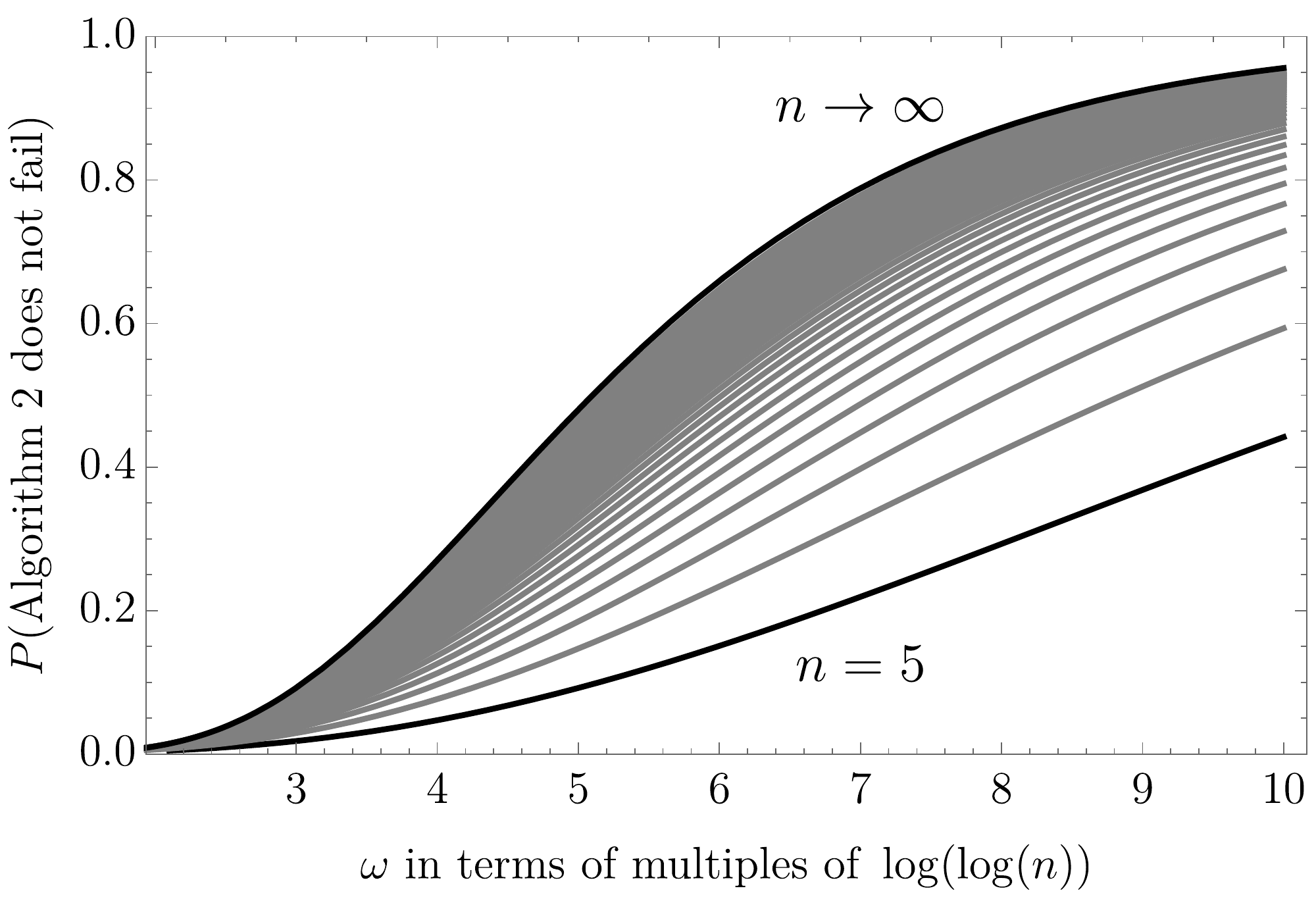}
    \caption{$\omega = 10(\log(\log(n))$ is a suitable choice for $n \geq 5$ to ensure that our overall quantum algorithm succeeds with a decent chance.}
    \label{fig:psucc}
\end{figure}

\end{document}